\documentclass[11pt,twoside]{article}
\usepackage{fullpage}

\usepackage{fullpage, amsthm, amsmath, amssymb, hyperref, dsfont, mathtools,bm}
\usepackage{framed}
\usepackage{caption,subcaption,graphicx,rotating,multirow}
\usepackage{hyperref,url}
\usepackage{color,xcolor}
\usepackage{enumitem}
\usepackage{epstopdf}
\usepackage[numbers]{natbib}

\theoremstyle{plain}
\newtheorem{theorem}{Theorem}

\newtheorem{lemma}[theorem]{Lemma}

\newtheorem{proposition}[theorem]{Proposition}
 \newenvironment{proofof}[1]{{\bf {\em Proof of #1.}}}{\hfill \rule{2mm}{2mm} %\qed 
 }
  
\newlength{\widebarargwidth}
\newlength{\widebarargheight}
\newlength{\widebarargdepth}
\DeclareRobustCommand{\widebar}[1]{%
  \settowidth{\widebarargwidth}{\ensuremath{#1}}%
  \settoheight{\widebarargheight}{\ensuremath{#1}}%
  \settodepth{\widebarargdepth}{\ensuremath{#1}}%
  \addtolength{\widebarargwidth}{-0.3\widebarargheight}%
  \addtolength{\widebarargwidth}{-0.3\widebarargdepth}%
  \makebox[0pt][l]{\hspace{0.3\widebarargheight}%
    \hspace{0.3\widebarargdepth}%
    \addtolength{\widebarargheight}{0.3ex}%
    \rule[\widebarargheight]{0.95\widebarargwidth}{0.1ex}}%
  {#1}}
  
\makeatletter
\long\def\@makecaption#1#2{
        \vskip 0.8ex
        \setbox\@tempboxa\hbox{\small {\bf #1:} #2}
        \parindent 1.5em  %% How can we use the global value of this???
        \dimen0=\hsize
        \advance\dimen0 by -3em
        \ifdim \wd\@tempboxa >\dimen0
                \hbox to \hsize{
                        \parindent 0em
                        \hfil 
                        \parbox{\dimen0}{\def\baselinestretch{0.96}\small
                                {\bf #1.} #2
                                } 
                        \hfil}
        \else \hbox to \hsize{\hfil \box\@tempboxa \hfil}
        \fi
        }
\makeatother

\newcommand{\normal}{\ensuremath{\mathcal{N}}}

\newcommand{\bigo}{\ensuremath{\mathcal{O}}}

\newcommand{\E}{\operatorname{\mathbb{E}}}

\newcommand{\eig}[1]{\ensuremath{\lambda_{#1}}}
\newcommand{\eigmax}{\ensuremath{\eig{\max}}}

\DeclareMathOperator{\Diag}{Diag}

\DeclareMathOperator{\trace}{trace}

\DeclareMathOperator{\rank}{rank}

\newcommand{\Zp}{\mathbb{Z}_{+}}
\newcommand{\R}{\mathbb{R}}

\newcommand{\Rp}{\mathbb{R}_{+}}

\newcommand{\indi}{\mathds{1}}

\newcommand{\imnb}{\mathbf{i}}

\newcommand{\twonorm}[1]{\left\|#1\right\|_{\ell_2}}

\newcommand{\vct}[1]{\bm{#1}}
\newcommand{\mtx}[1]{\bm{#1}}

\def\BC{\begin{center}}
\def\EC{\end{center}}
\def\BIT{\begin{itemize}}
\def\EIT{\end{itemize}}
\def\BET{\begin{enumerate}}
\def\EET{\end{enumerate}}
\def\BEQ{\begin{equation}}
\def\EEQ{\end{equation}}

\long\def\comment#1{}

\newcommand{\MSE}{\mathsf{MSE}}

\newcommand{\one}{\mathbf{1}}
\newcommand{\half}{\frac{1}{2}}

\newcommand{\symmpsd}{\mathbb{S}^{ d\times d}_{+}}
\renewcommand{\intercal}{{\bm \top}}

\newcommand{\BP}{\mathsf{BP}} % Belief Propagation
\newcommand{\HQP}{\mathsf{HQP}} % Histogram Queries Problem
\begin{document}

\title{
{\bf{\LARGE{Decoding from Pooled Data: \\ Phase Transitions of Message Passing}}}
}

\author{Ahmed El Alaoui\thanks{Department of Electrical Engineering and Computer Sciences, UC Berkeley, CA.} ~~ 
Aaditya Ramdas$^{\star}$\thanks{Department of Statistics, UC Berkeley, CA.} \and
Florent Krzakala\thanks{Laboratoire de Physique Statistique, CNRS, PSL Universit\'es \& Ecole Normale Sup\'erieure, Sorbonne Universit\'es et Universit\'e Pierre \& Marie Curie, Paris, France.} ~~ 
Lenka Zdeborov\'{a}\thanks{Institut de Physique Th\'eorique, CNRS, CEA, Universit\'e Paris-Saclay, Gif-sur-Yvette, France.} ~~ 
Michael I. Jordan$^{\star \dagger}$}

\date{}
\maketitle

\vspace*{-.3in} 
\begin{abstract}
 We consider the problem of decoding a discrete signal of categorical variables from the observation of several histograms of pooled subsets of it. 
 We present an Approximate Message Passing (AMP) algorithm for recovering the signal in the \emph{random dense} setting where each observed histogram involves a random subset of entries of size proportional to $n$. We characterize the performance of the algorithm in the asymptotic regime where the number of observations $m$ tends to infinity proportionally to $n$, by deriving the corresponding State Evolution (SE) equations and studying their dynamics. We initiate the analysis of the multi-dimensional SE dynamics by proving their convergence to a fixed point, along with some further properties of the iterates. The analysis reveals sharp phase transition phenomena where the behavior of AMP changes from exact recovery to weak correlation with the signal as $m/n$ crosses a threshold. We derive formulae for the threshold in some special cases and show that they accurately match experimental behavior.
\end{abstract}

\section{Introduction}
Consider a discrete high-dimensional signal consisting of categorical variables, 
for example, nucleotides in a string of DNA or country of origin for a set of people.  
In many real-world settings, it is infeasible to observe the entire high-dimensional 
signal, for reasons of cost or privacy.  Instead, in a manner akin to compressed 
sensing, observations can be obtained in the form of ``histograms'' or ``frequency 
spectra''---pooled measurements counting the occurence of each category or type 
across subsets of the variables. 
Concretely, we investigate the so-called \emph{Histogram Query Problem} ($\HQP$): a database consisting of a population of $n$ individuals, where each individual belongs to one category among $d$, is queried. In each query, a subset of individuals is selected, and the histogram of their types, along with the individuals in that subset are revealed. Such a data acquisition model is common in applications such as the processing of genetic data, where DNA samples from multiple sources are pooled and analyzed together~\cite{sham2002dna}. 
This gives rise to the inferential problem of determining the category of every individual in the population. 
 The question of interest in this paper is to determine the minimal number of observations needed for recovery, and to ascertain whether this inferential problem can be solved 
in an efficient manner.

\subsection{The setting}
Let $\tau^*  : \{1,\cdots,n\} \mapsto \{1,\cdots,d\}$ be an assignment of $n$ variables to $d$ categories. We denote the queried subpopulations by $S_a \subset \{1,\cdots,n\}$, $1\le a \le m$. Given $m$ subsets $S_a$, the histogram of categories of the pooled subpopulation $S_a$ is denoted by $\vct{h}_a \in \Zp^{d}$, i.e., for all $1\le a \le m$,
\begin{equation}\label{histogram_combinatorial}
\vct{h}_a := \left(\left|{\tau}^{*-1}(1) \cap S_a \right|,\cdots, \left|{\tau}^{*-1}(d) \cap S_a \right| \right).
\end{equation}
We let $\vct{\pi} = \frac{1}{n} \left(\left|{\tau^*}^{-1}(1)\right|,\cdots,\left|{\tau^*}^{-1}(d)\right| \right)$ denote the vector of proportions of assigned values; i.e., the empirical distribution of categories. We place ourselves in a random dense regime in which the sets $\{S_a\}_{1\le a \le m}$ are independent draws of a random set $S$ where $\Pr(i \in S) = \alpha$ independently for each $i \in \{1,\ldots,n\}$, for some fixed $\alpha\in (0,1)$. Meaning, at each query, the size of the pool is proportional to the size of the population: $\E[|S|] = \alpha n$.

Here we adopt a linear-algebraic formulation which will be more convenient for the presentation of the algorithm. We can represent the map $\tau^*$, which we refer to as \emph{the planted solution}, as a set of vectors $\vct{x}^*_i = \vct{e}_{\tau^*(i)} \in \R^d$, for $1\le i \le n$. Let $\mtx{A}\in \R^{m \times n}$ represent the sensing matrix: $A_{ai} = \indi\{i\in S_a\}$, for all $1\le i \le n, 1\le a \le m$.
The histogram equations~\eqref{histogram_combinatorial} can be written in the form of a linear system of $m$ equations:
 \begin{equation}\label{histogram_linear}
 \vct{h}_{a} = \sum_{i=1}^n A_{ai}\vct{x}^*_i,~\quad a \in \{1,\cdots,m\}.
 \end{equation}
Our goal can thus be rephrased as that of inverting the linear system~\eqref{histogram_linear}. Note that the problem becomes trivial if $m = n$, since the square random matrix $\mtx{A}$ will be invertible with high probability. However, as we review in the next section, a detailed information-theoretic analysis of the problem shows that the planted solution is uniquely determined by the above linear system for $m = \gamma \frac{n}{\log n}$, $\gamma >0$. In this paper we study the algorithmic problem in the regime $m = \kappa n$, $\kappa <1$.     

\subsection{Prior work}  
The $\HQP$ has recently been considered in~\cite{wang2016data,alaoui2016decoding}. Its study was initiated in~\cite{wang2016data} in the two settings where the sets $\{S_a\}$ are deterministic and random. We review the information-theoretic and algorithmic results known so far. 

\paragraph{Information-theoretic aspect} 
Under the condition that $\vct{\pi}$ is the uniform distribution, Wang~\emph{et al.}~\cite{wang2016data} showed a lower bound on the minimum number of queries $m$ for the problem to be well-posed, namely, if $m < \frac{\log d}{d-1}\frac{n}{\log n}$ then the set of collected histograms does not uniquely determine the planted solution $\tau^*$. Further, under the condition that $\alpha = \half$, they showed that $m > c_0 \frac{n}{\log n}$ with $c_0$ a constant independent of $d$, suffices to uniquely determine $\tau^*$. These results were later generalized and sharpened in~\cite{alaoui2016decoding}, where it was shown that for arbitrary $\vct{\pi}$ and $\alpha$, $m\in (\gamma_{\text{low}}\frac{n}{\log n}, \gamma_{\text{up}} \frac{n}{\log n})$ measurements are necessary and sufficient for $\tau^*$ to be unique, where $\gamma_{\text{low}} = \frac{H(\vct{\pi})}{d-1}$, and $\gamma_{\text{up}}$ is ``essentially" $2\gamma_{\text{low}}$ (see~\cite{alaoui2016decoding} for the precise formula), $H$ being the Shannon entropy function.    
    
\paragraph{Algorithmic aspect} In the deterministic setting, where one is allowed to design the sensing matrix $\mtx{A}$, i.e.\ choose the pools $S_a$ at each query, Wang~\emph{et al.}~\cite{wang2016data} provided a querying strategy that recovers $\tau^*$ provided that $m > c_1 \frac{n}{\log n}$, where $c_1$ is an absolute constant. Ignoring the dependence on $d$, this almost matches the information-theoretic limit. The random setting has not been treated so far, and is the subject of the present paper.    

\subsection{Contributions}
We present an Approximate Message Passing (AMP) algorithm for the \emph{random dense} setting, where each query involves a random subset of individuals of size proportional to $n$. We characterize the exact asymptotic behavior of the algorithm in the limit of large number of individuals $n$ and a proportionally large number of queries $m$, i.e.\ $m/n \to \kappa$. This is done by heuristically deriving the corresponding State Evolution (SE) equations corresponding to the AMP algorithm. Then, a rigorous analysis of the SE dynamics reveals a rich and interesting behavior; namely the existence of phase transition phenomena in the parameters $\kappa,d,\vct{\pi}$ of the problem, due to which the behavior of AMP changes radically, from exact recovery to  very weak correlation with the planted solution. We exactly locate these phase transitions in simple situations, such as the binary case $d=2$, the symmetric case $\vct{\pi} = (\frac{1}{d},\cdots,\frac{1}{d})$, and the general case under the condition that the SE iteration is initialized from a special point. The latter exhibits an intriguing phenomenon: the existence of not one, but an entire sequence of thresholds in the parameter $\kappa$ that rules the behavior of the SE dynamics. These thresholds correspond to sharp changes in the structure of the covariance matrix of the estimates output by AMP. We expect this phenomenon to be generic beyond the special initialization case studied here.
Beyond the precise characterization of the phase transition thresholds in these special cases, we initiate the study of State Evolution in a multivariate setting by proving the convergence of the full-dimensional SE iteration, when initialized from a ``far enough" point, to a fixed point, and show further properties of the iterate sequence. This paper is intended to be a sequel to the information-theoretic study conducted in ~\cite{alaoui2016decoding}.  

%%%%%%%%%%%%%%%%%%%%%%%%%%%%%%%%%%%%%%%
 \section{Approximate Message Passing and State Evolution}
In this section we present the Approximate Message Passing (AMP) algorithm and the corresponding State Evolution (SE) equations.     
 
\subsection{The AMP algorithm} \label{subsec:AMP}
The AMP algorithm~\cite{donoho2009message}, known as the Thouless-Anderson-Palmer equations in the statistical physics literature~\cite{thouless1977solution}, can be derived from Belief Propagation (BP) on the factor graph modeling the recovery problem. The latter is a bipartite graph of $n + m$ vertices. The variables $\{\vct{x}_i: 1\le i \le n\}$ constitute one side of the bipartition, and the observations $\{\vct{h}_a: 1\le a \le m\}$ constitute the other side. The adjacency structure is encoded in the sensing matrix $\mtx{A}$. Endowing each edge $(i,a)$ with two messages $\vct{m}_{i\to a}, \vct{m}_{a \to i} \in \Delta^{d-1}$, $\Delta^{d-1}$ being the probability simplex, one can write the self-consistency equations for the messages at each node by enforcing the histogram constraints at each observation (or check) node while treating the incoming messages as probabilistically independent in the marginalization operation. The iterative version of these self-consistency equations is the BP algorithm. BP is further simplified to AMP by exploiting the fact that the factor graph is random and dense, i.e.\ one only needs to track the average of the messages incoming to each node. This reduces the number of passed messages from $m\times n$ to $m+n$. For the present $d$-variate problem, the algorithm we present is a special case of Hybrid-GAMP of~\cite{rangan2012hybrid}. We let $\bar{\vct{h}}_a = (\vct{h}_a - \alpha n \vct{\pi})/\sqrt{n}$ and $\widebar{\mtx{A}} = (\mtx{A} - \alpha \one_m\one_n^\intercal)/\sqrt{n}$ be the centered and rescaled data, and assume that the parameters $\alpha$ and $\vct{\pi}$ are known to the algorithm. 
The AMP algorithm reads as follows:  At iteration $t = 1, 2, \dots$, we update the check nodes $a=1,\cdots,m$ as
\begin{align*}
\vct{\omega}_{a}^t  ~~&= \sum_{j \in \partial a} \widebar{A}_{aj} \hat{\vct{x}}_{j}^t - \mtx{V}_{a}^t\left(\mtx{V}_{a}^{t-1}\right)^{-1}(\bar{\vct{h}}_a - \vct{\omega}_a^{t-1}),\\
\mtx{V}_{a}^t ~~&= \sum_{j \in \partial a} \widebar{A}_{aj}^2\mtx{B}_{j}^t, 
\end{align*}
and then update the variable nodes $i = 1,\cdots, n$ as
\begin{align*}
\vct{z}_{i}^t ~~&= \hat{\vct{x}}_{i}^t + \mtx{\Sigma}_{i}^t \cdot \sum_{b \in \partial i} \widebar{A}_{bi} \left(\mtx{V}_{b}^t\right)^{-1} (\bar{\vct{h}}_b - \vct{\omega}^t_{b}), \\
\mtx{\Sigma}_{i}^t ~~&= \bigg(\sum_{b \in \partial i } \widebar{A}_{bi}^2 \left(\mtx{V}_{b}^t\right)^{-1}\bigg)^{-1}, \\
\hat{\vct{x}}_{i}^{t+1} &= \vct{\eta}(\vct{z}_{i}^t, \mtx{\Sigma}_{i}^t),\\ 
\mtx{B}_{i}^{t+1} &= \Diag(\hat{\vct{x}}_{i}^{t+1}) - \hat{\vct{x}}_{i}^{t+1} \cdot \hat{\vct{x}}_{i}^{t+1\intercal},
\end{align*}
with
\begin{equation}\label{eta_thresholding}
\vct{\eta}(\vct{z},\mtx{\Sigma}) : = \sum_{r=1}^d \pi_r \vct{e}_r\frac{ e^{-\half (\vct{z} - \vct{e}_r)^\intercal\mtx{\Sigma}^{-1}(\vct{z} - \vct{e}_r)}}{Z(\vct{z},\mtx{\Sigma})}  \in \R^d,
\end{equation}
where $Z(\vct{z},\mtx{\Sigma}) = \sum_{r=1}^d \pi_r e^{-\half (\vct{z} - \vct{e}_r)^\intercal\mtx{\Sigma}^{-1}(\vct{z} - \vct{e}_r)}$ is a normalization factor so that the entries of $\vct{\eta}$ sum to one. The map $\vct{\eta}$ plays the role of a ``thresholding function" with a matrix parameter $\mtx{\Sigma}$ that is adaptively tuned by the algorithm. One should compare this situation to the case of sparse estimation~\cite{donoho2009message} where the soft thresholding function is used. Here, the form taken by $\vct{\eta}$ is adapted to the structure of the signal we seek to recover. The variables $\vct{\omega}_{a}$ and $\mtx{V}_{a}$ represent estimates of the histogram $\vct{h}_a$ and their variances. The variables $\vct{z}_i$ and $\mtx{\Sigma}_i$ are estimators of the planted solution $\vct{x}_i^*$ and their variances before thresholding, while $\hat{\vct{x}}_i \in \Delta^{d-1}$ and $\mtx{B}_i$ are the posterior estimates of $\vct{x}^*_i$ and its variance, i.e., after thresholding. The algorithm can be initialized in a ``non-informative" way by setting $\hat{\vct{x}}_i^0 = \vct{\pi}, \mtx{B}_i^0 = \Diag(\vct{\pi}) - \vct{\pi} \vct{\pi}^\intercal$ for all $i = 1,\dots,n$, and $\vct{\omega}_a^{-1} = \vct{0}$ and $\mtx{V}_a^{-1} = \mtx{I}$ for all $a=1,\cdots,m$ for example.
 We defer the details of the derivation to Appendix~\ref{sxn:AMP}.

\subsection{State Evolution} 
State Evolution (SE)~\cite{donoho2009message,bayati2012universality}, known as the cavity method in statistical physics \cite{mezard1990spin}, allows us to exactly characterize the asymptotic behavior of AMP at each time step $t$, by tracking the evolution in time of the relevant \emph{order parameters} of the algorithm. More precisely, let 
\begin{align*}
\mtx{M}_{t,n} &:= \frac{1}{n} \sum_{i=1}^n \hat{\vct{x}}_i^{t} \vct{x}_i^{*\intercal},
 ~~~~ \text{and} ~~~~ \mtx{Q}_{t,n} := \frac{1}{n} \sum_{i=1}^n \hat{\vct{x}}_i^{t} \hat{\vct{x}}_i^{t\intercal}.
 \end{align*} 
The matrix $\mtx{M}_{t,n}$ tracks the average alignment of the estimates with the true solution, and $\mtx{Q}_{t,n}$ their average covariance structure.
The SE equations relate the values of these order parameters at $t+1$ to those at time $t$ in the limit $n \to \infty$, $m/n \to \kappa$. We let $\mtx{M}_{t}$ and $\mtx{Q}_{t}$ denote the respective limits of $\mtx{M}_{t,n}$ and $\mtx{Q}_{t,n}$, which we assume exist in this ``replica-symmetric" regime, and let $\mtx{D} = \text{Diag}(\vct{\pi})$. The SE equations read
\begin{align*}
\mtx{M}_{t+1} &= \sum_{r=1}^d \pi_r\E_{\vct{g}}\big[\vct{\eta}(\vct{e}_r + \mtx{X}_{t}^{\half}\vct{g}, \kappa^{-1}\mtx{R}_t)\big] \cdot \vct{e}_r^\intercal,\\
\mtx{Q}_{t+1} &= \sum_{r=1}^d \pi_r\E_{\vct{g}}\big[ \vct{\eta}(\vct{e}_r + \mtx{X}_{t}^{\half}\vct{g}, \kappa^{-1}\mtx{R}_t) 
 \cdot \vct{\eta}(\vct{e}_r + \mtx{X}_{t}^{\half}\vct{g}, \kappa^{-1}\mtx{R}_t)^\intercal \big],\\
\mtx{X}_{t}~~~ &= \kappa^{-1}(\mtx{D} - \mtx{M}_{t} - \mtx{M}_{t}^\intercal + \mtx{Q}_{t}),\\
\mtx{R}_t ~~~ &= \Diag(\mtx{Q}_t \one) - \mtx{Q}_t, 
\end{align*}
with $\vct{g} \sim \normal(\vct{0},\mtx{I})$. The matrix $\kappa \mtx{X}_{t}$ is the covariance matrix of the error of the estimates output by AMP at time $t$, and $\mtx{R}_t$ can be interpreted as the average covariance matrix of the estimates themselves. Note that the parameter $\alpha$ has disappeared from the characterization by the SE equations, just as in the information theoretic study~\cite{alaoui2016decoding}.  

The full derivation of these equations is relegated to Appendix~\ref{sxn:SE}. The main hypothesis behind the derivation, which we \emph{do not} rigorously verify, is that the variables $\vct{z}^t_i$ are asymptotically Gaussian, centered about $\vct{x}^*_i$ and with covariance $\mtx{X}_t$: the measure $\frac{1}{n}\sum_{i=1}^n \delta_{\vct{z}_{i}^t - \vct{x}^*_i}$ converges weakly to $\normal(\vct{0},\mtx{X}_t)$. We refer to~\cite{bayati2011dynamics,bayati2012universality} for rigorous results, the assumptions of which do not apply to this setting. It is an interesting problem to prove the exactness of the SE equations in this setting.

\subsection{Simplification of SE}
Here we simplify the system of SE equations above to a single iteration. 
 This crucially relies on the following Proposition:
\begin{proposition}\label{Nishimori}
If $\mtx{M}_0 = \mtx{Q}_0$, then for all $t$ we have 
\begin{itemize}
\item[(i)] $\mtx{M}_t = \mtx{Q}_t$. In particular, $\mtx{M}_t$ is a symmetric PSD matrix, and $\mtx{M}_t\one = \vct{\pi}$.
\item[(ii)] $\mtx{R}_t = \kappa \mtx{X}_t = \mtx{D} - \mtx{M}_t$. 
\end{itemize}  
\end{proposition}

The proof of the above proposition is deferred to Appendix~\ref{sxn:proofs}. We pause to make a few remarks. The assumption of the Proposition could be enforced for example by setting the initial estimates of AMP as $\hat{\vct{x}}_i^0 = \vct{\pi}$ for all $i$. This yields $\mtx{M}_0 = \mtx{Q}_0 = \vct{\pi} \vct{\pi}^\intercal$, and hence $\mtx{X}_0 = \kappa^{-1}(\mtx{D} - \vct{\pi}\vct{\pi}^\intercal)$. The statements in the Proposition ---together referred to as the \emph{Nishimori identities} in the statistical physics literature~\cite{zdeborova2015statistical}--- simplify the SE equations to a single iteration on $\mtx{X}_{t}$.  To succinctly present this simplification, for $r \in \{1,\cdots,d\}$, and $\mtx{X} \succeq \mtx{0}$, we let
\[\vct{\eta}_r(\mtx{X}) := \vct{\eta}(\vct{e}_r + \mtx{X}^{\half}\vct{g}, \mtx{X})  \in \Delta^{d-1}.\]
Then, the SE equations can be seen to boil down to the single equation
\begin{equation}\label{state_evolution}
\mtx{X}_{t+1} = \kappa^{-1}f(\mtx{X}_{t}),
\end{equation}
where, recalling that  $\vct{g} \sim \normal(\vct{0},\mtx{I})$, we define
\begin{align}\label{map_f_q}
f(\mtx{X}) &:= \mtx{D} -  \sum_{r=1}^d \pi_r\E_{\vct{g}}\big[\vct{\eta}_r(\mtx{X})\vct{\eta}_r(\mtx{X})^\intercal\big]\\
 &= \mtx{D} - \sum_{r=1}^d \pi_r \E_{\vct{g}}\left[\vct{\eta}_r(\mtx{X}) \right] \cdot \vct{e}_r^\intercal, \label{map_f_m}\\
&= \sum_{r=1}^d \pi_r \E_{\vct{g}}\left[(\vct{e}_r - \vct{\eta}_r(\mtx{X})) \cdot (\vct{e}_r - \vct{\eta}_r(\mtx{X}))^\intercal \right],\label{map_f_psd}
\end{align}
where equations \eqref{map_f_q} and \eqref{map_f_m} correspond to substituting the value of $\mtx{Q}_t$ and $\mtx{M}_t$ into statement (ii) of the above proposition, while the last equality~\eqref{map_f_psd} is just a consequence of the first two, \eqref{map_f_q} and \eqref{map_f_m}.  
Furthermore, via elementary algebra, the coordinates of the vector $\vct{\eta}_r(\mtx{X})$ can written as
\begin{equation}\label{eta_function}
\left(\vct{\eta}_r(\mtx{X})\right)_s = \frac{\pi_s \exp\left(-\vct{g}^\intercal \mtx{X}^{-\half}(\vct{e}_r - \vct{e}_{s})-\half\twonorm{\mtx{X}^{-\half}(\vct{e}_r - \vct{e}_{s})}^2\right)}{Z_r(\mtx{X})},
\end{equation}
with 
\[Z_r(\mtx{X}) :=  \sum_{s=1}^d \pi_{s} \exp\left(-\vct{g}^\intercal \mtx{X}^{-\half}(\vct{e}_r - \vct{e}_{s})-\half\twonorm{\mtx{X}^{-\half}(\vct{e}_r - \vct{e}_{s})}^2\right).\]

\subsection{The mean squared \& 0-1 errors} 
We can measure the performance of AMP by the mean squared error of the estimates $\{\hat{\vct{x}}_i^t\}_{i=1}^n$:
\[\MSE_{t,n} = \frac{1}{n} \sum_{i=1}^n \twonorm{\hat{\vct{x}}_i^t - \vct{x}^*_i}^2.\]
Since $\hat{\vct{x}}_i^t \in \Delta^{d-1}$, an alternative measure of performance would be the expected 0-1 distance between a random category drawn from the multinomial $\hat{\vct{x}}_i$ and the true category $\vct{x}_i^*$, then averaged over $i=1,\cdots,n$. This error would be written as
\begin{align*}
\frac{1}{n}\sum_{i=1}^n \sum_{r=1}^d &\hat{x}_{ir}^t(1-\vct{e}_r^\intercal \vct{x}^*_i) = 1 - \frac{1}{n} \sum_{i=1}^n\hat{\vct{x}}_{i}^{t\intercal} \vct{x}^*_i\\
 &= 1 - \trace(\mtx{M}_{t,n}) = \trace\left(\mtx{D} -\mtx{M}_{t,n} \right).
\end{align*}
On the other hand, the MSE in the large $n$ limit reads 
\begin{align*}
\MSE_{t} &:= \lim_{n \to \infty} \MSE_{t,n} 
= \trace\left(\mtx{Q}_{t} - \mtx{M}_{t} - \mtx{M}_{t}^\intercal +\mtx{D}\right),\\
&= \trace\left(\mtx{D} -\mtx{M}_{t} \right),
\end{align*}
so the two notions of error coincide in the limit. Note that the MSE at each step $t$ can be deduced from SE iterate at time $t$: $\MSE_{t} = \kappa \trace(\mtx{X}_{t})$.

\section{Analysis of the State Evolution dynamics}
In this section we present our main results on the convergence of the SE iteration~\eqref{state_evolution} to a fixed point, 
and the location of the phase transition thresholds  in three special cases. We start by analyzing the SE map $f$ and present some important generic results. 

\subsection{Analysis of the SE map $f$}

From expression~\eqref{map_f_psd}, we see that the map $f$ sends the positive semi-definite (PSD) cone $\symmpsd$ to itself. As written, $f$ is only defined for invertible matrices $\mtx{X}$, but it could be extended by continuity to singular matrices: if $\vct{e}_r -\vct{e}_s$ is in the null space of $\mtx{X}$, we declare that $\exp (-\half \|\mtx{X}^{-\half}(\vct{e}_r -\vct{e}_s)\|^2) = 0$. This convention is consistent with the limiting value of a sequence $\big\{\exp (-\half \|\mtx{X}_n^{-\half}(\vct{e}_r -\vct{e}_s)\|^2)\big\}_{n\ge 0}$ where $\big\{\mtx{X}_n\big\}_{n\ge 0}$ is a sequence of invertible matrices approaching $\mtx{X}$. This also has an interpretation based on an analogy with electrical circuits, which we discuss shortly.
This extension will be also denoted by $f$. It is continuous over the $\symmpsd$, and we have $f(\mtx{0}) = \mtx{0}$. Now, we state an important property of $f$, namely that it is monotone:
\begin{proposition} \label{monotonicity_f}
The map $f$ is order-preserving on $\symmpsd$; i.e., 
for all $\mtx{X},\mtx{Y} \succeq \mtx{0}$, if $\mtx{X} \preceq \mtx{Y}$ then $f(\mtx{X}) \preceq f(\mtx{Y})$. 
\end{proposition}
%\vspace{-.2cm}
The proof of this Proposition is conceptually simple but technical, and is thus deferred to Appendix~\ref{sxn:proofs}. Next, we adopt a combinatorial view of the structure of the SE dynamics. This will help us identity subspaces of $\symmpsd$ that are left invariant by $f$. Note that the definition of $f$ involves $\mtx{X}^{-\half}$ acting on $\text{span}(\one)^\perp$. Additionally, it is easy to verify that for all $\mtx{X} \in \symmpsd$, $f(\mtx{X})\one = \vct{0}$, and $f(\mtx{X})_{rs} \le 0$ for all $r \neq s$. Therefore, without loss of generality, we can restrict the study of the state evolution iteration to the set 
\[\mathcal{A} := \left\{\mtx{X} \in \symmpsd, ~\mtx{X}\one  = \vct{0}, ~X_{rs}\le 0 ~\forall (r,s) ~\text{s.t.}~r\neq s \right\},\] 
since it is invariant under the dynamics. 
The set $\mathcal{A}$ can be seen as the set of Laplacian matrices of weighted graphs on $d$ vertices (every edge $(r,s)$ is weighted by $-X_{rs}$ for $\mtx{X} \in \mathcal{A}$). 
Hence $f$ can be seen as a transformation on weighted graphs. This transformation enjoys the following invariance property: 
\begin{proposition}\label{block_diagonal}
For all $\mtx{X} \in \mathcal{A}$, $f$ preserves the connected component structure of the graph represented by $\mtx{X}$; i.e, two distinct connected components of the graph whose Laplacian matrix is $\mtx{X}$ remain distinct when transformed by $f$.
\end{proposition}
\begin{proof}
The proof relies on the concept of \emph{effective resistance}. One can view a graph of Laplacian $\mtx{X} \in \mathcal{A}$ as a network of resistors with resistances $1/(-X_{rs})$. The effective resistance of an edge $(r,s)$ is the resistance of the entire network when one unit of current is injected at $r$ and collected at $s$ (or vice-versa). Its expression is a simple consequence of Kirchhoff's law, and is equal to $R_{rs}:=\twonorm{\mtx{X}^{-1/2}(\vct{e}_r - \vct{e}_s)}^2$ (see e.g.\ \cite{spielman_sgt_course}). It is clear that the effective resistance of an edge is finite if and only if both its endpoints belong to the same connected component of the graph, otherwise $R_{rs} = +\infty$, and $(\vct{\eta}_r(\mtx{X}))_s = 0$. This causes $f$ to ``factor" across connected components, and thus acts on them independently.      
\end{proof}

Next, let us look at the limit of $f(t\mtx{X})$ for large $t$. For $\mtx{X} \in \mathcal{A}$ invertible on $\text{span}(\one)^\perp$, we have  
$\lim_{t \to \infty} f(t\mtx{X}) = \mtx{D} - \vct{\pi}\vct{\pi}^\intercal$, since $\vct{\eta}_r(t\mtx{X}) \to \vct{\pi}$ almost surely. More generally, if $\mtx{X}$ represents a graph with $\{V_k\}_{1\le k \le K}$ connected components, $(\vct{\eta}_r(t\mtx{X}))_s \neq 0$ only if $r,s$ are in the same component. Hence, $\vct{\eta}_r(t\mtx{X}) \to \frac{\mtx{P}_k \vct{\pi}}{\one^\intercal \mtx{P}_k \vct{\pi}}$, where $\mtx{P}_k$ is the orthogonal projector onto the span of the coordinates in $V_k$ where $r \in V_k$, and we have 
\begin{equation}\label{limit_def_L}
\lim_{t \to \infty} f(t\mtx{X}) = \mtx{D} - \sum_{k=1}^{K} \frac{\mtx{P}_k\vct{\pi}\vct{\pi}^\intercal\mtx{P}_k}{\one^\intercal \mtx{P}_k \vct{\pi}} =: \mtx{L}_K.\end{equation}
By Propositions~\ref{monotonicity_f} and~\ref{block_diagonal} and the limit calculation~\eqref{limit_def_L}, we deduce that for any partition $\{V_k\}_{1\le k \le K}$ of $\{1,\cdots,d\}$, and all Laplacian matrices $\mtx{X} \succeq \mtx{0}$ of graphs with connected components $V_1,\cdots,V_K$, we have 
\begin{equation}\label{upper_limit}
f(\mtx{X}) \preceq \mtx{L}_K.
\end{equation}  
Indeed, since $\mtx{X} \preceq t \mtx{X}$ for all $t \ge 1$, we have $f(\mtx{X}) \preceq  f(t\mtx{X})$ by monotonicity of $f$. Letting $t \to \infty$ settles the claim. In particular, with $K=1$, $\mtx{L}_1  =  \mtx{D} - \vct{\pi}\vct{\pi}^\intercal$, and for all $\mtx{X} \in \mathcal{A}$ representing a connected graph (i.e.\ $\rank(\mtx{X}) = d-1$), we have $f(\mtx{X}) \preceq \mtx{D} - \vct{\pi}\vct{\pi}^\intercal$.
We are now ready to state the main result of this subsection. 
\begin{theorem} 
Let $\{V_k\}_{1\le k \le K}$ be a partition of $\{1,\cdots,d\}$, and $\mtx{L}_K$ defined as in~\eqref{limit_def_L}. Let $\mtx{X}_0 \in \mathcal{A}$ with connected components $V_1,\cdots,V_K$, and such that $\mtx{X}_0 \succeq \kappa^{-1}\mtx{L}_K$. 
If the SE iteration~\eqref{state_evolution} is initialized from $\mtx{X}_0$, then the sequence $\left\{\mtx{X}_t\right\}_{t\ge 0}$ is decreasing in the PSD order, i.e., $\mtx{X}_{t} \preceq \mtx{X}_{t-1}$ for all $t\ge 1$, and converges to a fixed point $\mtx{X}^*$, i.e., $\mtx{X}^* = \kappa^{-1}f(\mtx{X}^*)$.
\end{theorem}
%\vspace{-.2cm}
\begin{proof}
Let $\mtx{X}_0$ satisfy the conditions of the Theorem. Using $\mtx{X}_0 \succeq \kappa^{-1}\mtx{L}_K$ and observation~\eqref{upper_limit}, we have $\mtx{X}_1 = \kappa^{-1}f(\mtx{X}_0) \preceq \mtx{X}_0$. By monotonicity of $f$, we deduce that the SE iterates form a monotone sequence: $\mtx{X}_{t+1} \preceq \mtx{X}_{t}$ for all $t \ge 0$. Since $\mtx{X}_t \succeq \mtx{0}$ for all $t$, then this sequence must have a limit\footnote{One can see this by observing that $\{\vct{z}^\intercal \mtx{X}_t \vct{z}\}_{t\ge 0}$ is a non-negative monotonically decreasing sequence for all $\vct{z} \in \R^d$; hence it must have a (non-negative) limit. Then, via the identity $\vct{y}^\intercal \mtx{X}_t \vct{z} = \half ((\vct{y}+\vct{z})^\intercal\mtx{X}_t (\vct{y}+\vct{z}) - (\vct{y}-\vct{z})^\intercal \mtx{X}_t (\vct{y}-\vct{z}))$, one deduces that $\{\vct{y}^\intercal \mtx{X}_t \vct{z}\}_{t \ge 0}$ has a limit for all $\vct{y},\vct{z} \in \R^d$. These limits define a bi-linear operator which is $(\vct{y},\vct{z}) \mapsto \vct{y}^\intercal \mtx{X}^* \vct{z}$.}  $\mtx{X}^* \succeq \mtx{0}$. By continuity of $f$, this limit must satisfy $\mtx{X}^* = \kappa^{-1}f(\mtx{X}^*)$.
\end{proof}
We expect that for $\kappa$ large enough, $\mtx{X}^* = \mtx{0}$, meaning that $\lim \mtx{M}_t = \mtx{D}$ and $\lim \MSE_t = 0$. This situation corresponds to perfect recovery of the planted solution $\{\vct{x}_i^*\}_{i=1}^n$ by AMP. We can easily show that this is the case for 
\begin{equation}\label{kappa_general}
\kappa > \kappa^* := \sup_{\mtx{X} \in \mathcal{A}} \frac{\eigmax(f(\mtx{X}))}{\eigmax(\mtx{X})}.
\end{equation}
Indeed,
\[\eigmax(\mtx{X}_{t+1}) = \kappa^{-1}\eigmax(f(\mtx{X}_t)) \le \frac{\kappa^*}{\kappa}\eigmax(\mtx{X}_t).\] 
If $\kappa > \kappa^*$ then the SE iterates converge to $\mtx{0}$ for \emph{every} initial point. It is currently unclear to us whether this condition is also necessary. Instead, we consider three special cases and exactly locate the phase transitions thresholds. 

%%%%%%%%%%%%%%%%%%%%%%%%%%%%%%%%%%%%%%%%%%%%%%%%%%%
\subsection{The binary case}
In this section we treat the case $d=2$, which is akin to a noiseless version of the CDMA problem~\cite{guo2005randomly} or the problem of compressed sensing with binary prior. In this case, the SE iteration becomes one-dimensional. Indeed, we have $\mathcal{A} = \left\{x\vct{u} \vct{u}^\intercal, x\ge 0\right\}$, with $\vct{u} = (1,-1)^\intercal$. And since this space is invariant under~$f$, the latter can be parameterized by one scalar function $x \mapsto \varphi(x)$, defined by  
\[ f(x\vct{u} \vct{u}^\intercal) = \varphi(x) \vct{u} \vct{u}^\intercal, \quad \forall x \ge 0.\]
Next, we compute $\varphi$. For $\mtx{X} = x\vct{u}\vct{u}^\intercal$, we have $\mtx{X}^{-\half}\vct{u} = \frac{1}{\sqrt{2x}} \vct{u}$.
Then, letting $\vct{\pi} = (p, 1-p)^\intercal$, using~\eqref{map_f_m} we have 
\begin{align} \label{calc_binary}
\varphi(x) = f(x\vct{u}\vct{u}^\intercal)_{1,1} 
&= p - p \E_{\vct{g}}\left[\frac{p}{p+(1-p)e^{-\vct{g}^\intercal\vct{u}/\sqrt{2x} - 1/2x}}\right],\nonumber\\
&= \E_{\vct{g}}\left[ \frac{p(1-p)}{ 1-p + pe^{\vct{g}^\intercal\vct{u}/\sqrt{2x} + 1/2x}}\right],\nonumber\\
&= \E_{g}\left[ \frac{p(1-p)}{ 1-p + pe^{g/\sqrt{x} +  1/2x}}\right].
\end{align}
Letting $\mtx{X}_t = a_t \vct{u} \vct{u}^\intercal$, for all $t\ge 0$, the SE reduces to 
\begin{equation}\label{SE_binary}
a_{t+1} = \kappa^{-1}\varphi(a_t).
\end{equation}
The function $\varphi$ is continuous, increasing on $\Rp$, and bounded (since $\varphi(\infty) = p(1-p) <\infty$). Moreover, $\varphi(0)=0$. Therefore, the sequence~\eqref{SE_binary} converges to zero for all initial conditions $a_0 >0$ if and only if $\kappa^{-1}\varphi(x) < x$ for all $x >0$, i.e.
\[\kappa > \kappa^*_{\text{binary}}(p) := \sup_{x>0}~ \E_{g}\left[ \frac{p(1-p)x^2}{ 1-p + p\exp\left(gx + x^2/2 \right)}\right].\]
By a change of variables $g + x/2 \to g$, one can also write this threshold as
\begin{align}\label{kappa_binary}
\kappa^*_{\text{binary}}(p) = \sup_{x>0}~\E_{g}\left[ \frac{p(1-p)x^2 e^{-x^2/8}}{pe^{gx/2} + (1-p)e^{-gx/2}}\right].
\end{align}
If $\kappa < \kappa^*_{\text{binary}}(p)$, then a new stable fixed point $a^*>0$ appears and the sequence $\{a_t\}_{t\ge 0}$ converges to it for all initial conditions $a_0 \ge a^*$, and the asymptotic MSE of the AMP algorithm is $\lim_{t\to \infty} \MSE_t = a^*\trace(\vct{u} \vct{u}^\intercal) = 2a^*$.

Figure~\ref{fig:binary} demonstrates the accuracy of the above theoretical predictions --- the predicted MSE by State Evolution matches the empirical MSE of AMP on a random instance with $n=2000$, across the whole range of $p$ and $\kappa$.
\begin{figure}
\centering 
\includegraphics[width=.43\textwidth]{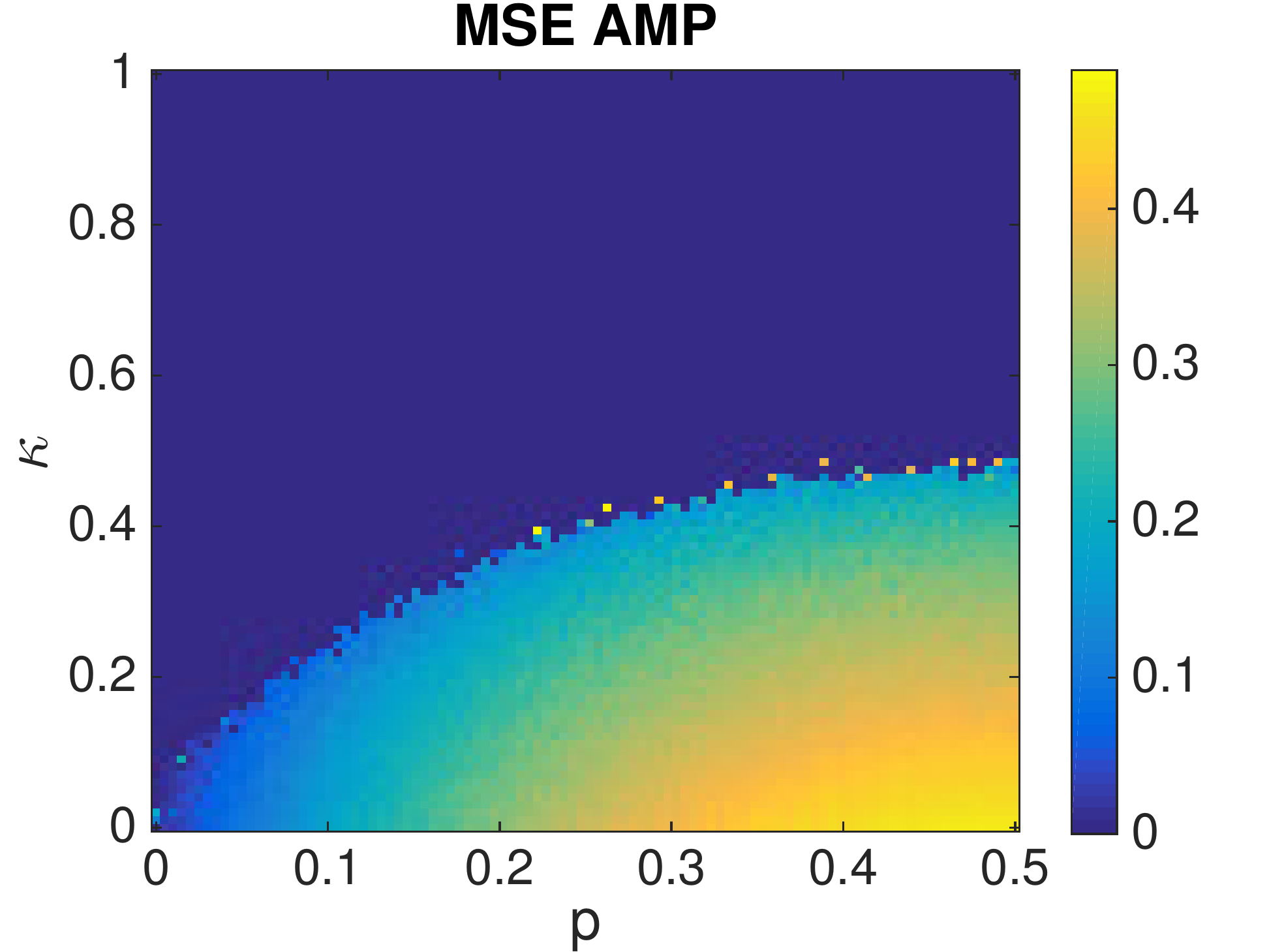}
\includegraphics[width=.43\textwidth]{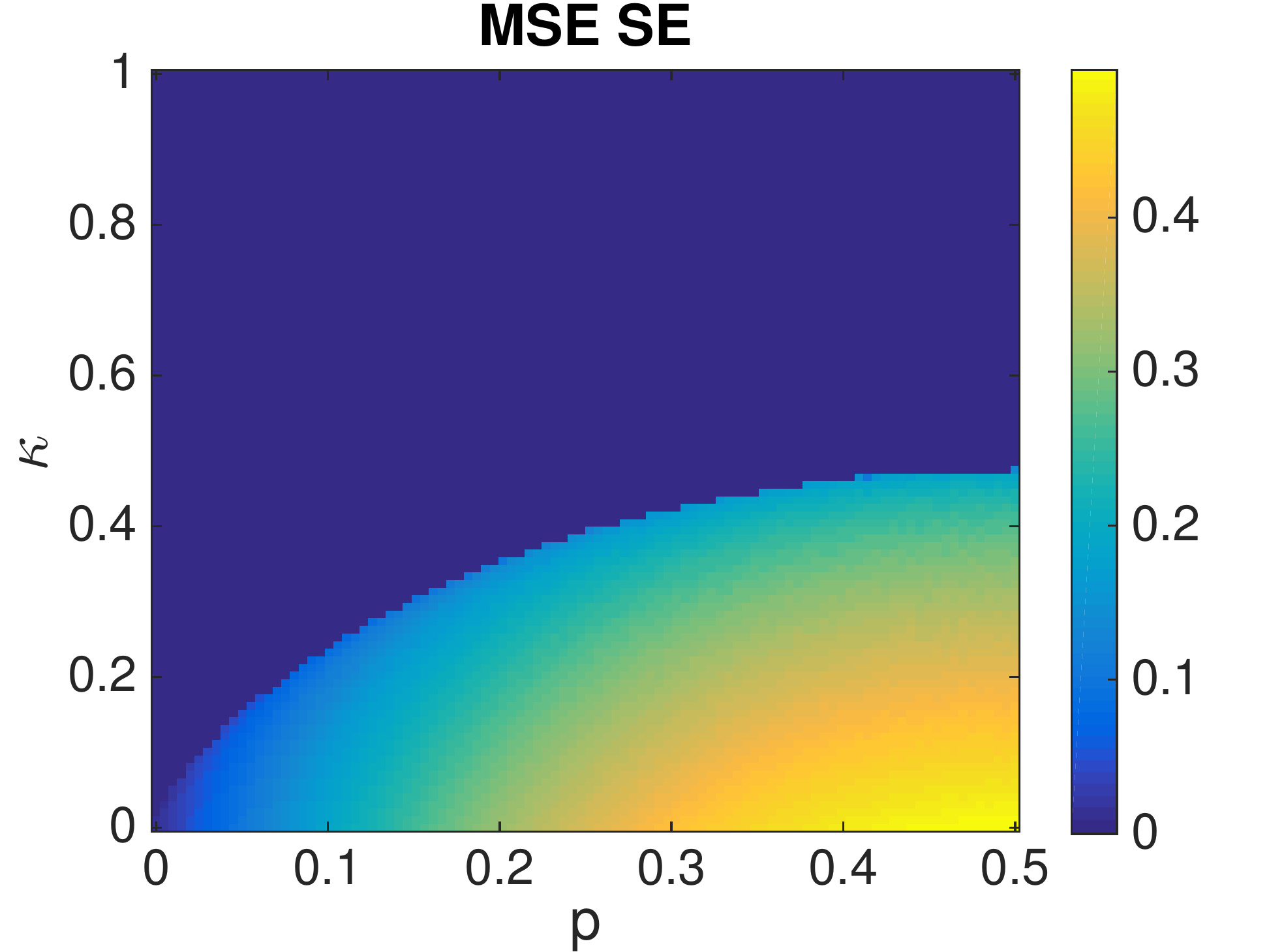}
\caption{MSE of AMP on a random instance with $n=2000$ in the binary case (left), and predicted MSE by State Evolution (right) as a function of $p = \pi_1$ and $\kappa$. The blue region corresponds to exact recovery. The boundary of this region is traced by the curve $p \mapsto \kappa^*_{\text{binary}}(p)$ in equation~\eqref{kappa_binary}.}\label{fig:binary}
\end{figure}
%%%%%%%%%%%%%%%%%%%%%%%%%%%%%%%%%%%%%%%%%%%%
\subsection{The symmetric case}
In this section we treat the symmetric case where all types have equal proportions: $\vct{\pi}  = (\frac{1}{d},\cdots,\frac{1}{d})$, and analyze the SE dynamics. In this situation, the half-line $\{ x (\mtx{D} - \vct{\pi}\vct{\pi}^\intercal)~,~ x \ge 0\}$ is stable under the application of the map $f$, and the dynamics becomes one-dimensional if initialized on this half-line.
\begin{lemma}  \label{lemma:symmetric}  
Assume $\vct{\pi}  = (\frac{1}{d},\cdots,\frac{1}{d})$. For all $x >0$, we have 
\[f\left( x (\mtx{I} - \frac{1}{d}\one \one^\intercal)\right) = \varphi(x)\left(\mtx{I} - \frac{1}{d}\one \one^\intercal\right),\]
with
\[\varphi(x) = \E_{\vct{g}}\left[\frac{\exp(g_2 /\sqrt{x} )}{\exp(g_1 /\sqrt{x} + 1/x) + \sum_{r=2}^d \exp(g_r /\sqrt{x} )}\right].\] 
%and $\vct{g} \sim \normal(\vct{0},\mtx{I})$.
\end{lemma}

\vspace{.3cm}
\begin{proof}
Let $\mtx{P} = (\mtx{I} - \frac{1}{d}\one \one^\intercal)$, and $\mtx{X} = x \mtx{P}$ with $x>0$. The matrix $\mtx{P}$ is the orthogonal projector on $\text{span}(\one)^\perp$. Therefore, we have 
\[\mtx{X}^{-1/2}(\vct{e}_r - \vct{e}_s) = (\vct{e}_r - \vct{e}_s)/\sqrt{x}.\]  
Therefore for all $r \neq s$, 
\[ f(\mtx{X})_{rs} = - \frac{1}{d}~\E_{\vct{g}}\left[ \frac{\exp\left(-\vct{g}^\intercal(\vct{e}_r - \vct{e}_s)/\sqrt{x} - 1/x\right)}{1 + \sum_{l\neq r} \exp\left(-\vct{g}^\intercal(\vct{e}_r - \vct{e}_l)/\sqrt{x} - 1/x \right)}\right].\]
By permutation-invariance of the Gaussian distribution, we see that $f(\mtx{X})$ is constant on its off-diagonal entries, hence on its diagonal entries as well since $f(\mtx{X})\one = \vct{0}$. Writing $f(\mtx{X}) = \frac{\alpha}{d}\mtx{I} - \frac{\beta}{d}( \one \one^\intercal - \mtx{I})$, we have $(\alpha+\beta) =d\beta$. Hence, $f(\mtx{X}) = \beta (\mtx{I} - \frac{1}{d}\one \one^\intercal)$ with 
\begin{align*}
\beta &= \E_{\vct{g}}\left[ \frac{\exp\left(-\vct{g}^\intercal(\vct{e}_1 - \vct{e}_2)/\sqrt{x} - 1/x\right)}{1 + \sum_{l\neq r} \exp\left(-\vct{g}^\intercal(\vct{e}_1 - \vct{e}_l)/\sqrt{x} - 1/x \right)}\right], \\
&= \E_{\vct{g}}\left[\frac{\exp(g_2 /\sqrt{x} )}{\exp(g_1 /\sqrt{x} + 1/x) + \sum_{r=2}^d \exp(g_r /\sqrt{x} )}\right],
\end{align*}
as claimed.
\end{proof}

Therefore, if the SE iteration is initialized on this half-line: $\mtx{X}_0 = a_0(\mtx{I} - \frac{1}{d}\one \one^\intercal)$, with $ a_0 >0$, then $\mtx{X}_t = a_t (\mtx{I} - \frac{1}{d}\one \one^\intercal)$ for all $t$, with 
\[a_{t+1} =  \kappa^{-1} \varphi(a_t).\]
Just as in the binary case, the function $\varphi$ is continuous, increasing and bounded with $\varphi(0)=0$. Hence, we have convergence to zero for all initial condition $a_0 >0$ if and only if $\kappa^{-1}\varphi(x) < x$ for all $x >0$, i.e. 
\begin{align}\label{kappa_symmetric}
\kappa > \kappa^*_{\text{sym}}(d) := \sup_{x>0} ~ \E_{\vct{g}}\left[\frac{x^2\exp\left(g_2 x\right)}{\exp\left(g_1x + x^2\right) + \sum_{r=2}^d \exp\left(g_r x\right)}\right].
\end{align}
Otherwise, it converges to a non-zero value $a^*$ for all initial conditions $a_0 > a^*$, and the asymptotic MSE of the AMP algorithm is $a^*\trace(\mtx{I} - \frac{1}{d}\one \one^\intercal) = (d-1)a^*$. Using the change of variables $g_1 + x \to g_1$, one can also write this threshold as
\[
\kappa^*_{\text{sym}}(d) = \sup_{x>0} ~  \E_{\vct{g}}\left[\frac{x^2 e^{-x^2/2}\exp((g_1+g_2) x)}{\sum_{r=1}^d \exp(g_r x)}\right]. 
\]
It is not straightforward to read off the magnitude of $\kappa^*_{\text{sym}}(d)$ from the above expression. We provide a table of approximate values for several small values of $d$: 
\begin{center}
\noindent
\begin{tabular}{|c|c|c|c|c|c|c|c|c|c|}
\hline
$d$ & 2 & 3 & 4 & 5 & 6 & 7 & 8 & 9 &10\\
\hline
$\kappa^*_{\text{sym}}$ & .47 & .39 & .34 & .30 & .27 & .24 & .22 & .21 & .20\\
\hline
\end{tabular}
\end{center}
For larger $d$, an asymptotic expression for this threshold may be desirable. We prove the following in Appendix~\ref{sxn:proofs}:
\begin{proposition} \label{asymptotics_large_d}
There exist two constants $0<c_l<c_u$ such that when $d$ is large enough,
\[c_l \frac{\log d}{d} \le \kappa^*_{\textup{sym}}(d) \le c_u \frac{\log d}{d},\] 
Furthermore, one can take $c_l = 1 - o_d(1)$, and $c_u = 2 + o_d(1)$.
\end{proposition}

%%%%%%%%%%%%%%%%%%%%%%%%%%%%%%%%%%%%%%%%%%%%%%
\subsection{The general case initialized with a matching}
Here we consider the SE iteration in arbitrary dimension and with arbitrary proportions of types $\vct{\pi}$, but we initialize the dynamics from a special point $\mtx{X}_0$ that corresponds to a matching of the vertices $\{1,\cdots,d\}$: each edge present in the matching corresponds to its own connected component. This case reveals an interesting behavior which we suspect is generic regardless of the initialization: the existence of a sequence of thresholds $\kappa_1^*,\kappa_2^*,\cdots$ ruling the behavior of the SE dynamics.       
Let $\mathcal{M} = \{(i_1,i_2),(i_3,i_4),\cdots,(i_{K-1},i_{K})\}$ be a matching on the set of vertices $\{1,\cdots,d\}$ (not all vertices are necessarily part of the matching), and let $\mtx{X}_0$ be its Laplacian matrix, where edges are weighted by arbitrary positive numbers. By Proposition~\ref{block_diagonal}, $f$ ``factors" across connected components, thus each edge in the matching will follow its own dynamics independently of the other edges. The edges not initially present in the matching remain inactive forever. For $(r,s) \in \mathcal{M}$, we have $(X_t)_{rr}= (X_t)_{ss} = - (X_t)_{rs} = - (X_t)_{sr}$, and 
\[\mtx{X}_t^{-\half}(\vct{e}_r-\vct{e}_s) = \frac{1}{\sqrt{2(X_t)_{rr}}}(\vct{e}_r-\vct{e}_s),\]
and therefore, using expression~\eqref{map_f_m} and letting $x  =(X_t)_{rr}$, 
\begin{align*}
f(\mtx{X}_t)_{rr} &= \pi_r-\E_{\vct{g}}\left[\left(\vct{\eta}_{r}(\mtx{X}_t)\right)_r\right],\\
&= \pi_r\E_{\vct{g}}\left[\frac{\pi_s }{\pi_r  e^{(g_r - g_s)/\sqrt{2x} + 1/2x} + \pi_s}\right], \\
&= \pi_r\E_{g}\left[\frac{\pi_s }{\pi_r  e^{g/\sqrt{x} + 1/2x} + \pi_s}\right].
\end{align*} 
Therefore, the SE iteration reduces to 
\[(X_{t+1})_{rr} = \kappa^{-1}\E_{g}\left[\frac{\pi_r\pi_s}{\pi_r  e^{g/\sqrt{(X_t)_{rr}} + 1/2(X_t)_{rr}} + \pi_s}\right],\]
for all vertices $(r,s) \in \mathcal{M}$. Note that this iteration is essentially the same as the one in the binary case~\eqref{calc_binary}-\eqref{SE_binary}, where $p$ becomes $\pi_r $ and $1-p$ becomes $\pi_s$.    
For each $(r,s) \in \mathcal{M}$, the above iteration converges to the fixed point zero for every initial point if and only if 
\begin{align}\label{kappa_matching}
\kappa > \kappa^*_{rs} := \sup_{x > 0} ~  \E_{g}\left[\frac{\pi_r\pi_s x^2e^{-x^2/8}}{\pi_r  e^{gx/2}+ \pi_s  e^{-gx/2}}\right].
\end{align}
Here, we symmetrized the expression just as in the binary case~\eqref{kappa_binary}.  Arranging these thresholds as $\kappa^*_1 > \kappa^*_2 > \cdots$ from largest to smallest we see that the fixed point of the SE iteration gains one non-zero edge at each $\kappa^*_i$ as $\kappa$ decreases from some large value to zero. Equivalently, $\mtx{X}^*$ gains a rank one component corresponding to the connected component constituted by that edge. It is an interesting problem to determine the behavior of the SE iteration and locate these thresholds, if they exist, beyond this simple matching case.

%%%%%%%%%%%%%%%%%%%%%%%%%%%%%%%%%%%%%%
\section{Conclusion}
We presented an algorithm for decoding categorical variables of a signal from randomly pooled observations of it, and characterized its performance it terms of a state evolution equation. The analysis of this evolution revealed phase transition phenomena in the parameters of the problem that happen in the linear regime $m/n \to \kappa$. These algorithmic results, combined with information-theoretic ones~\cite{wang2016data,alaoui2016decoding} leave a large region in parameter space ($\gamma \frac{n}{\log n} < m < \kappa n$) where the signal is identifiable but AMP fails at recovering it, hinting at a possible computational hardness in this structured signal recovery problem. This could have interesting applications in privacy-related considerations. Further, we proved the convergence of the SE dynamics to a fixed point. The analysis of the properties of this fixed point as a function of the parameters $\kappa, \vct{\pi}$ in the general case, together with rigorous proof of the exactness of the state evolution equations for this problem are interesting open problems.    

%\vspace{-.1cm}
\paragraph{Acknowledgment}
Part of this work was performed 
when FK and LZ were visiting the Simons Institute for the Theory of Computing at UC Berkeley. 
FK  acknowledges funding from the EU (FP/2007-2013/ERC grant agreement 307087-SPARCS). 
MJ acknowledges the support of the Mathematical Data Science program of the Office of Naval Research under grant number N00014-15-1-2670.

%\begin{small}

\bibliographystyle{alpha}
\bibliography{histograms}

%\end{small}

\appendix

\section{Omitted proofs}
 \label{sxn:proofs}

\subsection{Proof of Proposition~\ref{Nishimori}}
We proceed by induction. 
Now assume that $\mtx{M}_{t-1} = \mtx{Q}_{t-1}$ and that $\mtx{R}_{t-1} = \kappa \mtx{X}_{t-1}$. We prove  that $\mtx{R}_{t} = \kappa \mtx{X}_{t}$ and then that $\mtx{M}_{t}$ is symmetric and $\mtx{M}_{t} = \mtx{Q}_{t}$. 

The first step is to show that $\mtx{Q}_{t} \one = \vct{\pi}$. By assumption, $\mtx{X}_{t-1} = \kappa^{-1}(\mtx{D} - \mtx{Q}_{t-1}) = \kappa^{-1} \mtx{R}_{t-1}$.  
Let us define,
\begin{equation}\label{eta}
\eta_{rs} := \vct{\eta}_r(\mtx{X})_s = \frac{\pi_s\exp\left(-\vct{g}^\intercal \mtx{X}^{-1/2}(\vct{e}_r - \vct{e}_{s})-\half\twonorm{\mtx{X}^{-1/2}(\vct{e}_r - \vct{e}_{s})}^2\right)}{Z_r(\mtx{X})}.
\end{equation} 

The $s$th coordinate of $\mtx{Q}_{t} \one$ is 
\[(\mtx{Q}_{t} \one)_s = \sum_{r=1}^d \pi_r\E_{\vct{g}}\left[\left(\vct{\eta}(\vct{e}_r + \mtx{X}_t^{1/2}\vct{g}, \kappa^{-1}\mtx{R}_t)\right)_s\right] = \sum_{r=1}^d \pi_r\E_{\vct{g}}\left[\eta_{rs}\right].\]
Moreover, letting $\vct{\delta}_{rs} = \mtx{X}_{t-1}^{-1/2}(\vct{e}_r -\vct{e}_s)$, we have
\begin{align*}
\E_{\vct{g}}\left[\eta_{rs}\right] 
&= \int \frac{\pi_s\exp(-\half\twonorm{\vct{g} + \vct{\delta}_{rs}}^2)}{\sum_{l=1}^d \pi_{l}\exp(-\half\twonorm{\vct{g} + \vct{\delta}_{rl}}^2)} \frac{e^{-\half\twonorm{\vct{g}}^2}}{(2\pi)^{d/2}} \mathrm{d}\vct{g},\\
&\stackrel{(i)}{=} \int \frac{\exp(-\half\twonorm{\vct{g}}^2)}{\sum_{l=1}^d \pi_{l}\exp(-\half\twonorm{\vct{g} + \vct{\delta}_{rl}}^2)} 
% \quad \quad \quad \quad \quad\times 
\frac{\pi_s e^{-\half\twonorm{\vct{g} - \vct{\delta}_{rs}}^2}}{(2\pi)^{d/2}} \mathrm{d}\vct{g},\\
&= \int \frac{\pi_s\exp(-\half\twonorm{\vct{g} + \vct{\delta}_{sr}}^2)}{\sum_{l=1}^d \pi_{l}\exp(-\half\twonorm{\vct{g} + \vct{\delta}_{sl}}^2)} \frac{e^{-\half\twonorm{\vct{g}}^2}}{(2\pi)^{d/2}} \mathrm{d}\vct{g}.
\end{align*}
The only non-trivial equality is $(i)$ and it was obtained through a simple change of variable $\vct{g} + \vct{\delta}_{rs}\to \vct{g}$. Therefore,
\[(\mtx{Q}_{t} \one)_s = \pi_s \sum_{r=1}^d \E_{\vct{g}}\left[\frac{\pi_r\exp(-\half\twonorm{\vct{g} + \vct{\delta}_{sr}}^2)}{\sum_{l=1}^d \pi_{l}\exp(-\half\twonorm{\vct{g} + \vct{\delta}_{sl}}^2)}\right] = \pi_s.\]
In addition, the above argument also shows that $\mtx{M}_{t}$ is symmetric since, for $r,s \in \{1,\cdots,d\}$, 
\[\left(\mtx{M}_{t}\right)_{rs} = \pi_s \E_{\vct{g}}[\eta_{sr}] .\]

Now we have that $\mtx{R}_{t} = \mtx{D} - \mtx{Q}_t$, and by symmetry of $\mtx{M}_t$, $\mtx{X}_t = \kappa^{-1}(\mtx{D} - 2\mtx{M}_t + \mtx{Q}_t)$. To complete the proof, it remains to show that $\mtx{M}_t = \mtx{Q}_t$. For $r,s \in \{1,\cdots,d\}$ we have
\begin{align*}
\left(\mtx{Q}_t\right)_{rs} = \sum_{l=1}^d \pi_l ~\E_{\vct{g}}\left[ \eta_{lr}\eta_{ls}\right].
\end{align*}
Once again, we make the change of variable $\vct{g} + \vct{\delta}_{lr} \to \vct{g}$:
\begin{align*}
\left(\mtx{Q}_t\right)_{rs}&= \pi_r \pi_s \sum_{l=1}^d \pi_l ~\int \frac{\exp(-\half\twonorm{\vct{g}}^2) ~ \exp(-\half\twonorm{\vct{g}+\vct{\delta}_{rs}}^2)}{\left(\sum_{l'=1}^d \pi_{l'}\exp(-\half\twonorm{\vct{g} + \vct{\delta}_{rl'}}^2)\right)^2} \frac{e^{-\half\twonorm{\vct{g} -\vct{\delta}_{lr}}^2}}{(2\pi)^{d/2}} \mathrm{d}\vct{g},\\
&= \pi_r \pi_s \E_{\vct{g}}\left[\frac{\exp(-\half\twonorm{\vct{g} + \vct{\delta}_{rs}}^2)}{\sum_{r'=1}^d \pi_{r'}\exp(-\half\twonorm{\vct{g} +\vct{\delta}_{rl'}}^2)} \right],\\
&= \left(\mtx{M}_t\right)_{sr}.
\end{align*}
%\end{proofof}

\subsection{Proof of Proposition~\ref{monotonicity_f}}
The map $f$ is differentiable at every $\mtx{X} \succeq \mtx{0}$ invertible on $\text{span}(\one)^\perp$. Let $ \mtx{0} \preceq \mtx{X} \preceq \mtx{Y}$, and $\mtx{W}: [0,1] \to \symmpsd$ defined by $\mtx{W}(t) = (1-t)\mtx{X} + t\mtx{Y}$. We will show that $ \frac{\mathrm{d}}{\mathrm{d}t}f(\mtx{W}(t)) \succeq \mtx{0}$ for all $t \in [0,1]$ and conclude with the fundamental theorem of calculus 
\[f(\mtx{Y}) - f(\mtx{X}) = \int_0^1 \frac{\mathrm{d}}{\mathrm{d}t}f(\mtx{W}(t))\mathrm{d}t.\]
We start by computing the derivative of each entry of $f(\mtx{W}(t))$. Let $r, s \in \{1,\dots,d\}$. We have
\[
\frac{\mathrm{d}}{\mathrm{d}t}f(\mtx{W}(t))_{rs} = - \frac{\mathrm{d}}{\mathrm{d}t} \pi_r \E\left[\left(\vct{\eta}_r(\mtx{W}(t))\right)_s\right].
\]
To prepare for further calculations, let us write 
\[\mtx{A}(t) := \mtx{W}(t)^{-1/2}\frac{\mathrm{d}}{\mathrm{d}t}\left(\mtx{W}(t)^{-1/2}\right),\] 
and 
\[\mtx{B}(t) := \frac{\mathrm{d}}{\mathrm{d}t}\left(\mtx{W}(t)^{-1}\right) = - \mtx{W}(t)^{-1}\cdot \frac{\mathrm{d}}{\mathrm{d}t}\mtx{W}(t) \cdot\mtx{W}(t)^{-1}.\] 
We observe by the chain rule of differentiation that 
\begin{align}\label{a_plus_at_eq_b}
\mtx{A}(t) + \mtx{A}(t)^\intercal = \mtx{B}(t).
\end{align}
This identity will be used several times. Now we start computing the derivative. Let
\begin{align*}
D_{rs} :&= \pi_s\frac{\mathrm{d}}{\mathrm{d}t} \exp\left(-\vct{g}^\intercal \mtx{W}(t)^{-1/2}(\vct{e}_r - \vct{e}_{s})-\frac{1}{2}\twonorm{\mtx{W}(t)^{-1/2}(\vct{e}_r - \vct{e}_{s})}^2\right) \\
&= \pi_s \left(-\vct{g}^\intercal \frac{\mathrm{d}}{\mathrm{d}t}\left(\mtx{W}(t)^{-1/2}\right) (\vct{e}_r - \vct{e}_s) - \half (\vct{e}_r - \vct{e}_s)^\intercal\mtx{B}(t)(\vct{e}_r - \vct{e}_s) \right) \\
&~~~~~ \times \exp\left(-\vct{g}^\intercal \mtx{W}(t)^{-1/2}(\vct{e}_r - \vct{e}_{s})-\frac{1}{2}\twonorm{\mtx{W}(t)^{-1/2}(\vct{e}_r - \vct{e}_{s})}^2\right).
\end{align*}
Then,
\begin{align}\label{derivative_eta}
\frac{\mathrm{d}}{\mathrm{d}t} \vct{\eta}_r(\mtx{W}(t))_s = \frac{D_{rs}}{Z_r(\mtx{W}(t))} - \vct{\eta}_r(\mtx{W}(t))_s \times \sum_{l=1}^d \frac{D_{rl}}{Z_r(\mtx{W}(t))}.
\end{align}
Now, by differentiating under the expectation sign, we are lead to process expressions of the form
\[\E_{\vct{g}}\left[\frac{D_{rs}}{Z_r(\mtx{W}(t))}\right] ~~~ \text{and} ~~~ \E_{\vct{g}}\left[\vct{\eta}_r(\mtx{W}(t))_s \frac{D_{rl}}{Z_r(\mtx{W}(t))}\right].\]
Here, the Gaussian integration by parts formula
\[\E_{g}\left[gh(g)\right] = \E_{g}\left[h'(g)\right]\]   
for all univariate differentiable functions $h$ with moderate growth (say polynomial) at infinity, will be used multiple times. Recalling 
\begin{equation*}
\eta_{rs} = \vct{\eta}_r(\mtx{W}(t))_s = \frac{\pi_s\exp\left(-\vct{g}^\intercal \mtx{W}(t)^{-1/2}(\vct{e}_r - \vct{e}_{s})-\frac{1}{2}\twonorm{\mtx{W}(t)^{-1/2}(\vct{e}_r - \vct{e}_{s})}^2\right)}{Z_r(\mtx{W}(t))},
\end{equation*} 
from~\eqref{eta}, we have
\begin{align*}
\E_{\vct{g}}\left[\vct{g}^\intercal \frac{\mathrm{d}}{\mathrm{d}t}\left(\mtx{W}(t)^{-1/2}\right) (\vct{e}_r - \vct{e}_s) ~\eta_{rs}\right]
&= \E_{\vct{g}}\left[\left(\nabla_{\vct{g}}\eta_{rs}\right)^\intercal \frac{\mathrm{d}}{\mathrm{d}t}\left(\mtx{W}(t)^{-1/2}\right) (\vct{e}_r - \vct{e}_s)\right]\\
&= -(\vct{e}_r - \vct{e}_{s})^\intercal \mtx{A}(t)(\vct{e}_r - \vct{e}_{s})~\E_{\vct{g}}\left[\eta_{rs}\right] \\
&~~  + \sum_{l=1}^d (\vct{e}_r - \vct{e}_{l})^\intercal \mtx{A}(t)(\vct{e}_r - \vct{e}_{s})~\E_{\vct{g}}\left[\eta_{rs}\eta_{rl}\right], 
\end{align*}
and similarly,
\begin{align*}
\E_{\vct{g}}\left[\vct{g}^\intercal \frac{\mathrm{d}}{\mathrm{d}t}\left(\mtx{W}(t)^{-1/2}\right) (\vct{e}_r - \vct{e}_l)~ \eta_{rs} \eta_{rl}\right]
&= -\big((\vct{e}_r - \vct{e}_{l})^\intercal \mtx{A}(t)(\vct{e}_r - \vct{e}_{l}) \\
&~~~~~\quad + (\vct{e}_r - \vct{e}_{s})^\intercal \mtx{A}(t)(\vct{e}_r - \vct{e}_{l})\big) \E_{\vct{g}}\left[\eta_{rs}\eta_{rl}\right] \\
&~~  + 2\sum_{r'=1}^d (\vct{e}_r - \vct{e}_{r'})^\intercal \mtx{A}(t)(\vct{e}_r - \vct{e}_{l})\E_{\vct{g}}\left[\eta_{rs} \eta_{rl}\eta_{rr'}\right].
\end{align*}
Therefore,
\begin{align*}
\E_{\vct{g}}\left[\frac{D_{rs}}{Z_r(\mtx{W}(t))}\right] 
&= (\vct{e}_r - \vct{e}_{s})^\intercal \mtx{A}(t)(\vct{e}_r - \vct{e}_{s})\E_{\vct{g}}\left[\eta_{rs}\right] -  \half (\vct{e}_r - \vct{e}_s)^\intercal\mtx{B}(t)(\vct{e}_r - \vct{e}_s)\E_{\vct{g}}\left[\eta_{rs}\right] \\
&~~  - \sum_{l=1}^d (\vct{e}_r - \vct{e}_{l})^\intercal \mtx{A}(t)(\vct{e}_r - \vct{e}_{s})~\E_{\vct{g}}\left[\eta_{rs}\eta_{rl}\right].
\end{align*}
Since $\mtx{A}(t) + \mtx{A}(t)^\intercal = \mtx{B}(t)$ (identity~\eqref{a_plus_at_eq_b}), the first two terms in the above expression cancel each other, and we are left with
\[\E_{\vct{g}}\left[\frac{D_{rs}}{Z_r(\mtx{W}(t))}\right] = - \sum_{l=1}^d (\vct{e}_r - \vct{e}_{l})^\intercal \mtx{A}(t)(\vct{e}_r - \vct{e}_{s})~\E_{\vct{g}}\left[\eta_{rs}\eta_{rl}\right].\]
On the other hand, using the identity~\eqref{a_plus_at_eq_b} again,
\begin{align*}
\E_{\vct{g}}\left[\vct{\eta}_r(\mtx{W}(t))_s \frac{D_{rl}}{Z_r(\mtx{W}(t))}\right] 
&= \big((\vct{e}_r - \vct{e}_{l})^\intercal \mtx{A}(t)(\vct{e}_r - \vct{e}_{l}) + (\vct{e}_r - \vct{e}_{s})^\intercal \mtx{A}(t)(\vct{e}_r - \vct{e}_{l})\big) \E_{\vct{g}}\left[\eta_{rs}\eta_{rl}\right] \\
&~~-\half (\vct{e}_r - \vct{e}_l)^\intercal\mtx{B}(t)(\vct{e}_r - \vct{e}_l)\E_{\vct{g}}\left[\eta_{rs}\eta_{rl}\right]\\
&~~ - 2\sum_{r'=1}^d (\vct{e}_r - \vct{e}_{r'})^\intercal \mtx{A}(t)(\vct{e}_r - \vct{e}_{l})\E_{\vct{g}}\left[\eta_{rs} \eta_{rl}\eta_{rr'}\right]\\
&= (\vct{e}_r - \vct{e}_{s})^\intercal \mtx{A}(t)(\vct{e}_r - \vct{e}_{l}) \E_{\vct{g}}\left[\eta_{rs}\eta_{rl}\right] \\
&~~ - 2\sum_{r'=1}^d (\vct{e}_r - \vct{e}_{r'})^\intercal \mtx{A}(t)(\vct{e}_r - \vct{e}_{l})\E_{\vct{g}}\left[\eta_{rs} \eta_{rl}\eta_{rr'}\right].
\end{align*} 
Now, using the above two formulas, and recalling~\eqref{derivative_eta}, we have 
\begin{align*}
\frac{\mathrm{d}}{\mathrm{d}t}\E\left[\vct{\eta}_r(\mtx{W}(t))_s\right] &= - \sum_{l=1}^d (\vct{e}_r - \vct{e}_{l})^\intercal \mtx{A}(t)(\vct{e}_r - \vct{e}_{s})\E_{\vct{g}}\left[\eta_{rs}\eta_{rl}\right] \\
&~~ -\sum_{l=1}^d (\vct{e}_r - \vct{e}_{s})^\intercal \mtx{A}(t)(\vct{e}_r - \vct{e}_{l}) \E_{\vct{g}}\left[\eta_{rs}\eta_{rl}\right]\\
&~~ +  2\sum_{l=1}^d\sum_{r'=1}^d (\vct{e}_r - \vct{e}_{r'})^\intercal \mtx{A}(t)(\vct{e}_r - \vct{e}_{l})\E_{\vct{g}}\left[\eta_{rs} \eta_{rl}\eta_{rr'}\right].
\end{align*}
Using identity~\eqref{a_plus_at_eq_b}, the sum of the first two terms in the above expression is 
\begin{align*}
&-\sum_{l=1}^d (\vct{e}_r - \vct{e}_{s})^\intercal \mtx{B}(t)(\vct{e}_r - \vct{e}_{l}) \E_{\vct{g}}\left[\eta_{rs}\eta_{rl}\right],\\
&=-\sum_{l,r'=1}^d (\vct{e}_r - \vct{e}_{s})^\intercal \mtx{B}(t)(\vct{e}_r - \vct{e}_{l}) \E_{\vct{g}}\left[\eta_{rs}\eta_{rl}\eta_{rr'}\right],
\end{align*}
where we used the fact $\sum_{r'} \eta_{rr'} = 1$ in the last expression. 
Similarly, the third term is equal to
\[\sum_{l,r'=1}^d (\vct{e}_r - \vct{e}_{r'})^\intercal \mtx{B}(t)(\vct{e}_r - \vct{e}_{l})\E_{\vct{g}}\left[\eta_{rs}\eta_{rl} \eta_{rr'}\right].\]
Therefore we obtain
\begin{align*}
\frac{\mathrm{d}}{\mathrm{d}t}\E\left[\vct{\eta}_r(\mtx{W}(t))_s\right] 
&= \sum_{l,r'=1}^d (\vct{e}_r - \vct{e}_{r'})^\intercal \mtx{B}(t)(\vct{e}_s - \vct{e}_{l})\E_{\vct{g}}\left[\eta_{rs}\eta_{rl} \eta_{rr'}\right].
\end{align*}  
The expression we just obtained does not appear to be symmetric in the indices $(r,s)$, but it does become symmetric when multiplied by $\pi_r$, thanks to the following identity:
\begin{lemma}\label{symmetrization}
Recall the definition of $\eta_{rs}$ from~\eqref{eta}. For all $r,s,l\in \{1,\cdots,d\}$ we have
\[\pi_r \E_{\vct{g}}\left[\eta_{rs}\eta_{rl}\right] = \sum_{l'=1}^d \pi_{l'} \E_{\vct{g}}\left[\eta_{l'r}\eta_{l's}\eta_{l'l}\right].\]
\end{lemma}   
Using the above, we get
\begin{align*}
\frac{\mathrm{d}}{\mathrm{d}t}f(\mtx{W}(t))_{rs} &= -\pi_r\frac{\mathrm{d}}{\mathrm{d}t}\E\left[\vct{\eta}_r(\mtx{W}(t))_s\right],\\ 
&= -\sum_{l,l',r'=1}^d \pi_{l'}(\vct{e}_r - \vct{e}_{r'})^\intercal \mtx{B}(t)(\vct{e}_s - \vct{e}_{l})\E_{\vct{g}}\left[\eta_{l'r}\eta_{l's} \eta_{l'l}\eta_{l'r'}\right],\\
&= -\sum_{l'=1}^d \pi_{l'}  \E_{\vct{g}}\left[\eta_{l'r}\eta_{l's} \cdot (\vct{e}_r - \vct{\eta}_{l'})^\intercal \mtx{B}(t)(\vct{e}_s - \vct{\eta}_{l'}) \right].
\end{align*}
This implies that for all $\vct{z} \in \R^d$
\begin{align*}
\vct{z}^\intercal \frac{\mathrm{d}}{\mathrm{d}t}f(\mtx{W}(t))\vct{z} 
&= \sum_{r,s=1}^d \frac{\mathrm{d}}{\mathrm{d}t}f(\mtx{W}(t))_{rs} z_rz_s,\\
&= -\sum_{l'=1}^d \pi_{l'} \E_{\vct{g}}\left[(\vct{z}\odot\vct{\eta}_{l'} - (\vct{z}^\intercal \vct{\eta}_{l'})\vct{\eta}_{l'})^\intercal \mtx{B}(t) (\vct{z}\odot\vct{\eta}_{l'} - (\vct{z}^\intercal \vct{\eta}_{l'})\vct{\eta}_{l'}) \right],
\end{align*}
where $\odot$ denote the entry-wise product of two vectors.
Since $\mtx{B}(t) = - \mtx{W}(t)^{-1}\cdot \frac{\mathrm{d}}{\mathrm{d}t}\mtx{W}(t) \cdot\mtx{W}(t)^{-1}$ and $\frac{\mathrm{d}}{\mathrm{d}t}\mtx{W}(t) = \mtx{Y} - \mtx{X} \succeq \mtx{0}$, we see that
\[\vct{z}^\intercal \frac{\mathrm{d}}{\mathrm{d}t}f(\mtx{W}(t))\vct{z}
= \sum_{l'=1}^d \pi_{l'} \E_{\vct{g}}\left[\twonorm{ (\mtx{Y} - \mtx{X})^{\half} \mtx{W}(t)^{-1} \left(\vct{z}\odot\vct{\eta}_{l'} - (\vct{z}^\intercal \vct{\eta}_{l'})\vct{\eta}_{l'}\right)}^2 \right] \ge 0,\]
hence concluding the general argument. It now remains to prove Lemma~\ref{symmetrization}. 

\vspace{.3cm}
 \noindent\begin{proofof}{Lemma~\ref{symmetrization}}
The proof relies on a simple change of variables in the expectation. Using~\eqref{eta}, and letting $\vct{\delta}_{rs} = \mtx{W}(t)^{-1/2}(\vct{e}_r - \vct{e}_{s})$ for all $r,s$, we have 
\begin{align*}
\E_{\vct{g}}\left[\eta_{l'r}\eta_{l's}\eta_{l'l}\right] &= \pi_r \pi_s \pi_l \E_{\vct{g}}\left[\frac{e^{-\vct{g}^\intercal (\vct{\delta}_{l'r}+\vct{\delta}_{l's}+\vct{\delta}_{l'l})-\half\twonorm{\vct{\delta}_{l'r}}^2-\half\twonorm{\vct{\delta}_{l's}}^2-\half\twonorm{\vct{\delta}_{l'l}}^2}}{\left(\sum_{r'=1}^d \pi_{r'} e^{-\vct{g}^\intercal  \vct{\delta}_{l'r'}-\half\twonorm{\vct{\delta}_{l'r'}}^2}\right)^3}\right]\\
&= \pi_r \pi_s \pi_l \E_{\vct{g}}\left[\frac{e^{-\half \twonorm{\vct{g} + \vct{\delta}_{l'r}}^2 -\half \twonorm{\vct{g} + \vct{\delta}_{l's}}^2 -\half \twonorm{\vct{g} + \vct{\delta}_{l'l}}^2}}{\left(\sum_{r'=1}^d \pi_{r'} e^{-\half \twonorm{\vct{g} + \vct{\delta}_{l'r'}}^2}\right)^3}\right]\\
&= \pi_r \pi_s \pi_l \int_{\R^d} \frac{e^{-\half \twonorm{\vct{g} + \vct{\delta}_{l'r}}^2 -\half \twonorm{\vct{g} + \vct{\delta}_{l's}}^2 -\half \twonorm{\vct{g} + \vct{\delta}_{l'l}}^2}}{\left(\sum_{r'=1}^d \pi_{r'} e^{-\half \twonorm{\vct{g} + \vct{\delta}_{l'r'}}^2}\right)^3} ~\frac{e^{-\half \twonorm{\vct{g}}^2}}{(2\pi)^{d/2}} ~\mathrm{d}\vct{g}.
\end{align*}
We make the change of variables $\vct{g} + \vct{\delta}_{l'r} \to \vct{g}$. The term $\twonorm{\vct{g} + \vct{\delta}_{l'r}}^2$ becomes $\twonorm{\vct{g}}^2$, $\twonorm{\vct{g} + \vct{\delta}_{l's}}^2$ becomes $\twonorm{\vct{g} + \vct{\delta}_{rs}}^2$,  $\twonorm{\vct{g} + \vct{\delta}_{l'l}}^2$ becomes $\twonorm{\vct{g} + \vct{\delta}_{rl}}^2$, $\twonorm{\vct{g}}^2$ becomes $\twonorm{\vct{g} + \vct{\delta}_{rl'}}^2$, and $\twonorm{\vct{g} + \vct{\delta}_{l'r'}}^2$ becomes $\twonorm{\vct{g} + \vct{\delta}_{rr'}}^2$ in the denominator. The first term will assume the role of the Gaussian density, and we rewrite the above as an expectation under the Gaussian distribution:
\[\pi_r \pi_s \pi_l \E_{\vct{g}}\left[\frac{e^{-\half \twonorm{\vct{g} + \vct{\delta}_{rs}}^2 -\half \twonorm{\vct{g} + \vct{\delta}_{rl}}^2 -\half \twonorm{\vct{g} + \vct{\delta}_{rl'}}^2}}{\left(\sum_{r'=1}^d \pi_{r'} e^{-\half \twonorm{\vct{g} + \vct{\delta}_{rr'}}^2}\right)^3}\right].\] 
If the above expression is multiplied by $\pi_{l'}$ and summed over all $l'$, the third term in the numerator cancels with one power of the denominator, and the result is 
\[\pi_r \pi_s \pi_l \E_{\vct{g}}\left[\frac{e^{-\half \twonorm{\vct{g} + \vct{\delta}_{rs}}^2 -\half \twonorm{\vct{g} + \vct{\delta}_{rl}}^2}}{\left(\sum_{r'=1}^d \pi_{r'} e^{-\half \twonorm{\vct{g} + \vct{\delta}_{rr'}}^2}\right)^2}\right] = \pi_r \E_{\vct{g}}\left[ \eta_{rs} \eta_{rl}\right].\] 
\end{proofof}

\subsection{Proof of Proposition~\ref{asymptotics_large_d}}
For $x >0$, we let
\[\phi_d(x)  := \E_{\vct{g}}\left[\frac{x^2 \sum_{r=2}^d e^{g_r \sqrt{\log(d-1)} x }}{e^{g_1\sqrt{\log(d-1)}x} \cdot (d-1)^{x^2} + \sum_{r=2}^d e^{g_r\sqrt{\log(d-1)} x }}\right].\]
By symmetry in the variables $g_r$, $r \ge 2$, we can see that
\[\phi_d\left(\frac{x}{\sqrt{\log(d-1)}}\right) =  \frac{d-1}{\log (d-1)}\E_{\vct{g}}\left[\frac{x^2\exp(g_2 x )}{\exp(g_1x + x^2) + \sum_{r=2}^d \exp(g_r x )}\right].\]
Our claim reduces to exhibiting upper and lower bounds on $\sup_{x>0} \phi_d(x)$ which are asymptotically independent of $d$.
We start with the upper bound. Since, the function $x \to \frac{x}{1+x}$ is concave on $\Rp$, we have by Jensen's inequality,
\begin{align*}
\phi_d(x)  &\le \E_{g_1}\left[\frac{x^2 \sum_{r=2}^d  \E_{g_r:r\ge 2}\left[e^{g_r \sqrt{\log(d-1)} x }\right]}{e^{g_1\sqrt{\log(d-1)}x} \cdot (d-1)^{x^2} + \sum_{r=2}^d   \E_{g_r:r\ge 2}\left[e^{g_r \sqrt{\log(d-1)} x }\right]}\right],\\
&= \E_{g_1}\left[\frac{x^2 (d-1)^{1+x^2/2} }{e^{g_1\sqrt{\log(d-1)}x} \cdot (d-1)^{x^2} + (d-1)^{1+x^2/2} }\right].
\end{align*}
We split the analysis into two cases: $x \le \sqrt{2}+ \epsilon$, $x > \sqrt{2}+ \epsilon$ for some $\epsilon >0$. We see that $\phi_d(x)  \le x^2$ for all $x > 0$. If $x \le \sqrt{2}+ \epsilon$, then $\phi_d(x) \le (\sqrt{2}+ \epsilon)^2$. For the remaining case, let $\alpha = \alpha(\epsilon) >0$ such that $x^2/2 - \alpha x -1 > 0$ for all $x > \sqrt{2}+\epsilon$. One can find such an $\alpha$ as the solution to the equation $\alpha + \sqrt{\alpha^2 + 2} = \sqrt{2}+ \epsilon$. 
Next, we let $\mathcal{E}$ be the event that $g_1 \le \frac{1-x^2/2+\alpha x}{x} \sqrt{\log(d-1)}$, and write
\begin{align*} 
\phi_{d}(x) &\le \E_{g_1}\left[\frac{x^2}{(d-1)^{x^2/2-1} \cdot e^{g_1\sqrt{\log(d-1)}x} + 1} \bigg| \bar{\mathcal{E}}\right] \Pr(\bar{\mathcal{E}})\\
&~~~+  \E_{g_1}\left[\frac{x^2}{(d-1)^{x^2/2-1} \cdot e^{g_1\sqrt{\log(d-1)}x} + 1} \bigg| \mathcal{E}\right] \Pr(\mathcal{E}),
\end{align*}  
Under $\bar{\mathcal{E}}$, we have $-x^2/2 +1 - g_1 x \sqrt{\log(d-1)} \le -\alpha x$, and the first term in the above expression is upper bounded by
\[x^2 (d-1)^{-\alpha x}.\]
On the other hand, we upper bound the conditional expectation in the second term by $x^2$, and use the fact that $\Pr(\mathcal{E}) \le (d-1)^{-(1-x^2/2+\alpha x)^2/(2x^2)}$. We obtain the upper bound
\[\phi_d(x) \le x^2 \left((d-1)^{-\alpha x} + (d-1)^{-(1-x^2/2+\alpha x)^2/(2x^2)}\right),\]
which decays to $0$ as $d \to \infty$ \underline{uniformly in $x \ge \sqrt{2}+\epsilon$}. This proves that 
\[\sup_{x >0}\phi_{d}(x) \le (\sqrt{2}+\epsilon)^2\]
 for all $d$ sufficiently large.  

Now we turn our attention on the lower bound. Since the function $x \to \frac{x}{1+x}$ is increasing, we have
\[\phi_d(x) \ge \E_{\vct{g}}\left[\frac{x^2  e^{\max_{r \ge 2}g_r \sqrt{\log(d-1)} x }}{e^{g_1\sqrt{\log(d-1)}x} \cdot (d-1)^{x^2} +e^{\max_{r \ge 2}g_r\sqrt{\log(d-1)} x }}\right].\]  
The maximum of finitely many Gaussians concentrates in a sub-Gaussian way: for all $t \ge 0$,
\[\Pr\left(\max_{r\ge 2} g_r - \E[\max_{r\ge 2} g_r] \le - t \right) \le e^{-t^2/2}.\]
We write $ \E[\max_{r\ge 2} g_r] = c_d \sqrt{\log(d-1)}$; it is known that $c_d = \sqrt{2}(1-o_d(1))$. Letting $t = \epsilon c_d \sqrt{\log(d-1)}$ for some $\epsilon>0$, we have 
\[\phi_d(x) \ge \E_{g_1}\left[\frac{x^2  (d-1)^{(1-\epsilon)c_d x}}{e^{g_1\sqrt{\log(d-1)}x} \cdot (d-1)^{x^2} + (d-1)^{(1-\epsilon)c_d x}}\right] \cdot \left(1-(d-1)^{-\epsilon^2 c_d^2/2}\right).\]
We plug the value $x = (1-\epsilon)c_d$ in the right hand side, and deduce
\[\sup_{x>0}\phi_d(x) \ge \E_{g_1}\left[\frac{(1-\epsilon)^2c_d^2}{e^{g_1(1-\epsilon)c_d \sqrt{\log(d-1)}} + 1}\right] \cdot \left(1-(d-1)^{-\epsilon^2 c_d^2/2}\right).\]
We see that the above converges to the value $(1-\epsilon)^2$ as $d\to \infty$.

%%%%%%%%%%%%%%%%%%%%%%%%%%%%%%%%%%%%%%%%%%%%%%
%  BP - AMP

\section{Deriving the Approximate Message Passing equations}
 \label{sxn:AMP}

 We divide the derivation of the AMP equations into two parts. First, we write down the Belief Propagation (BP) equations, and simplify them to the ``relaxed'' BP equations. Then, we show how to transform the relaxed BP equations into the AMP iteration.

\subsection{From Belief Propagation (BP) to Relaxed BP} 
The factor graph $G$ of our model consists of a bipartite graph with the variables $\{\vct{x}_i,~ 1\le i \le n\}$ on one side of the bipartition and the measurements $\{\vct{h}_a,~ 1\le a \le m\}$ on the other side. A measurement (or check) node $\vct{h}_a$ is connected to $k = \alpha n$ variables nodes in expectation chosen uniformly at random (i.e. those such that $A_{ai}=1$) from all the variable nodes. 

We rescale the elements of the sensing matrix $\mtx{A}$ such that $A_{ai}$ has expectation 0 and variance $\frac{\alpha(1-\alpha)}{n}$. This can be done by subtracting the vector $\alpha n\vct{\pi}$ from each observation $\vct{h}_a$ and dividing everything by $\sqrt{n}$. Hence, we let 
\[\widebar{\vct{h}}_a := (\vct{h}_a - \alpha n \vct{\pi})/\sqrt{n},\] 
and 
\[\widebar{\mtx{A}} = (\mtx{A} - \alpha \one_m \one_n^\intercal)/\sqrt{n}.\] 
The linear system $\vct{h}_a = \sum_{j=1}^n A_{aj} \vct{x}^*_j$ is equivalent to $\bar{\vct{h}}_a = \sum_{j=1}^n \widebar{A}_{aj}\vct{x}^*_j$. 

We now write the messages of the Belief Propagation algorithm. Let $\vec{E}$ be the set of directed edges of the factor graph with all possible directions, i.e., each edge $(i,a)$ is endowed with both directions $i\to a$ and $a \to i$. Note that $|\vec{E}| = 2 k m$. The message passing procedure consists of iterating a map $\BP: \left(\Delta^{d-1}\right)^{\vec{E}} \to \left(\Delta^{d-1}\right)^{\vec{E}}$ from some initial guess until (possible) convergence. 
For convenience, for all $r \in \{1,\cdots,d\}$, any set of messages $ \vct{m} = \{\vct{m}_{i\to a}~,~\vct{m}_{a\to i}~:~  A_{ai}=1\} \in \left(\Delta^{d-1}\right)^{\vec{E}}$ on $G$, and any directed edge $a \to i$, we denote the $r$th coordinate of the $d$-dimensional message $\vct{m}_{a \to i}$ by $\vct{m}_{a \to i}(\vct{e}_r)$ instead of $(\vct{m}_{a \to i})_r$, and similarly for the coordinates of $\vct{m}_{a\to i}$.
With this notation in hand, the map $\BP$ is defined as follows: We consider a prior distribution on the messages that agrees with the category proportions in the planted solution $\tau^*$, i.e., for every $i$ and $r$,
\[P(\vct{x}_i = \vct{e}_r) = \pi_r\]
This is our ``uninformative" prior: under lack of any further information, the algorithm predicts that $\vct{x}_i = \vct{e}_r$ with probability $\pi_r$ for all $i$ and $r$. 
Then for all $\vct{x} \in \{\vct{e}_1,\cdots,\vct{e}_d\}$,
\begin{align}
 \BP(\vct{m})_{i \to a}(\vct{x}) &:= \frac{1}{Z_{i \to a}(\vct{m})} P(\vct{x}) \prod_{b \in \partial i \backslash a} \vct{m}_{b \to i}(\vct{x}),\label{var_comp}\\
 \BP(\vct{m})_{a \to i}(\vct{x}) &:=  \frac{1}{Z_{a \to i}(\vct{m})}\sum_{\vct{x}_j \in \{\vct{e}_1,\cdots,\vct{e}_d\} \atop j \in \partial a \backslash i} \indi \left\{\bar{\vct{h}}_a = \widebar{A}_{ai}\vct{x}+\sum_{ j  \neq i} \widebar{A}_{aj}\vct{x}_j \right\} \prod_{j \in \partial a \backslash i} \vct{m}_{j \to a}(\vct{x}_j),\label{check_comp}
 \end{align}
 with $Z_{i \to a}(\vct{m})$ and $Z_{a \to i}(\vct{m})$ are the normalizing factors such that $\sum_{r=1}^d\BP(\vct{m})_{i \to a}(\vct{e}_r) = \sum_{r=1}^d\BP(\vct{m})_{a \to i}(\vct{e}_r) = 1$. If $G$ was a tree, the map $\BP$ would compute the exact posterior distribution of the category assignments $\{\vct{x}_i~:~ 1 \le i \le n\}$ given the observations $\{\vct{h}_a~:~ 1 \le a \le m\}$. In our case this will only be true when $m/n$ is large enough.   

We see that the second equation above has a sum involving $d^{k-1}$ terms, which makes the execution of the BP algorithm intractable. We derive a set of \emph{relaxed} Belief Propagation messages from the above that only require linear-algebraic computations of size polynomial in $n$ and $m$. Later, we further simplify these equations by leveraging the fact that our factor graph is random and dense, to finally arrive at the Approximate Message Passing iteration.

We now proceed by replacing the indicator in~\eqref{check_comp} by a Gaussian with small variance $\sigma>0$, which we then linearize by writing it as the Fourier transform of the standard Gaussian measure (this is also known as the Hubbard-Stratonovich transformation):     
 \begin{align*}
\BP_{\sigma}(\vct{m})_{a \to i}(\vct{x}) &:= \frac{1}{Z_{a \to i}(\vct{m})}
\sum_{ \vct{x}_j \in \{\vct{e}_1,\cdots,\vct{e}_d\} \atop j \in \partial a \backslash i} 
\exp
\bigg(-\bigg\|\bar{\vct{h}}_a - \sum_{j =1}^n \widebar{A}_{aj}\vct{x}_j\bigg\|_{\ell_2}^2\bigg/2\sigma^2 \bigg) 
\prod_{j \in \partial a \backslash i} \vct{m}_{j \to a}(\vct{x_j}), \\
&\propto 
\sum_{\vct{x}_j \in \{\vct{e}_1,\cdots,\vct{e}_d\} \atop j \in \partial a \backslash i} 
\int_{\R^d} 
\exp\bigg( \imnb \sigma^{-1} \vct{g}^\intercal \bigg(\bar{\vct{h}}_a - \sum_{j =1}^n \widebar{A}_{aj}\vct{x}_j\bigg) \bigg) 
\prod_{j \in \partial a \backslash i} \vct{m}_{j \to a}(\vct{x}_j) \gamma_d(\mathrm{d}\vct{g}),\\
&\text{where we let $\gamma_d$ refer to the standard $d$-dimensional Gaussian measure.}\\
&\propto
\int_{\R^d} 
\exp\bigg(\imnb \sigma^{-1} \vct{g}^\intercal \big(\bar{\vct{h}}_a - \widebar{A}_{ai} \vct{x}\big) \bigg) \\
& \quad\quad\quad\times 
\prod_{j \in \partial a \backslash i}
\bigg[ \sum_{\vct{x}_j \in \{\vct{e}_1,\cdots,\vct{e}_d\}} \exp\big(-\imnb \sigma^{-1} \widebar{A}_{aj} \vct{g}^\intercal \vct{x}_j \big) ~\vct{m}_{j \to a}(\vct{x}_j) \bigg]
\gamma_d(\mathrm{d}\vct{g}).
\end{align*}

Now, observe that the exponentials in the sum above involve the terms $\widebar{A}_{aj}$ which are of order $1/\sqrt{n}$. By expanding the Taylor series of the exponential, one can show 
\begin{align*}
\sum_{\vct{x}_j \in \{\vct{e}_1,\cdots,\vct{e}_d\}} \exp(-\imnb \sigma^{-1}  \widebar{A}_{aj} \vct{g}^\intercal \vct{x}_j )~\vct{m}_{j \to a}(\vct{x}_j) &=
\sum_{r =1}^d \exp(-\imnb \sigma^{-1}  \widebar{A}_{aj} \vct{g}^\intercal \vct{e}_r )~\vct{m}_{j \to a}(\vct{e}_r)\\
&= 
 \exp\bigg(-\imnb \sigma^{-1}  \widebar{A}_{aj} \vct{g}^\intercal \vct{m}_{j \to a} - \half \sigma^{-2} \widebar{A}_{aj}^2 \vct{g}^\intercal \mtx{B}_{j \to a} \vct{g}\bigg) \\
 &\hspace{1cm}+ \bigo(1/n^{3/2}),
\end{align*}
where 
\begin{equation}\label{B_ja}
\mtx{B}_{j \to a}  = \Diag( \vct{m}_{j \to a}) - \vct{m}_{j \to a}\vct{m}_{j \to a}^\intercal.
\end{equation}
Plugging the above expression into the message, we get
\begin{align*}
\BP_{\sigma}(\vct{m})_{a \to i}(\vct{x}) &\approx \frac{1}{Z_{a \to i}(\vct{m})}  
\int_{\R^d} 
\exp\bigg(\imnb \sigma^{-1}  \vct{g}^\intercal \big(\bar{\vct{h}}_a - \widebar{A}_{ai} \vct{x}\big) \bigg) \\
& \quad\quad\quad\times 
\prod_{j \in \partial a \backslash i}
\exp\bigg(-\imnb \sigma^{-1}  \widebar{A}_{aj} \vct{g}^\intercal \vct{m}_{j \to a} - \half \sigma^{-2} \widebar{A}_{aj}^2 \vct{g}^\intercal \mtx{B}_{j \to a} \vct{g}\bigg) 
\gamma_d(\mathrm{d}\vct{g}),\\
 &=  \frac{1}{Z_{a \to i}(\vct{m})}  
 \int_{\R^d} 
\exp\bigg(
\imnb \sigma^{-1}  \vct{g}^\intercal \big(\bar{\vct{h}}_a - \widebar{A}_{ai} \vct{x}\big) \\
&\hspace{2.5cm}
-\sum_{j \in \partial a \backslash i} \imnb \sigma^{-1}  \widebar{A}_{aj} \vct{g}^\intercal \vct{m}_{j \to a} 
- \half \sum_{j \in \partial a \backslash i} \sigma^{-2} \widebar{A}_{aj}^2 \vct{g}^\intercal \mtx{B}_{j \to a} \vct{g}
\bigg) 
\gamma_d(\mathrm{d}\vct{g}).
\end{align*}
We denote the ``average" message and variance that appear in the formula above by
\begin{align}
\vct{\omega}_{a \to i}  &:= \sum_{j \in \partial a \backslash i} \widebar{A}_{aj} \vct{m}_{j \to a},\label{avg_check_to_var}\\
\mtx{V}_{a \to i} &:= \sum_{j \in \partial a \backslash i} \widebar{A}_{aj}^2 \mtx{B}_{j \to a}. \label{variance_check_to_var}
\end{align}
The exponentiated term in the integrand, when combined with the contribution of the Gaussian density, becomes 
\begin{align*}
\imnb \sigma^{-1}\vct{g}^\intercal \big(\bar{\vct{h}}_a - \widebar{A}_{ai} \vct{x} - \vct{\omega}_{a \to i} \big)& - \half \sigma^{-2}\vct{g}^\intercal \mtx{V}_{a \to i} \vct{g}
-\half \twonorm{\vct{g}}^2\\
&= \imnb \sigma^{-1}\vct{g}^\intercal \big(\bar{\vct{h}}_a - \widebar{A}_{ai} \vct{x} - \vct{\omega}_{a \to i} \big)- \half \vct{g}^\intercal (\sigma^{-2}\mtx{V}_{a \to i} + \mtx{I})\vct{g}.
\end{align*}
 Now, computing the integral yields
 \[\BP_{\sigma}(\vct{m})_{a \to i}(\vct{x}) \propto \exp\left(- \frac{1}{2\sigma^2} \twonorm{(\sigma^{-2}\mtx{V}_{a \to i} + \mtx{I})^{-\half}\big(\bar{\vct{h}}_a - \widebar{A}_{ai} \vct{x} - \vct{\omega}_{a \to i} \big)}^2 \right),  \] 
 and letting $\sigma \to 0$ yields
 \begin{align*}
 \BP(\vct{m})_{a \to i}(\vct{x}) \propto \exp\left(- \half \twonorm{\mtx{V}_{a \to i}^{-\half}\big(\bar{\vct{h}}_a - \widebar{A}_{ai} \vct{x} - \vct{\omega}_{a \to i} \big)}^2 \right).  
 \end{align*}
On the other hand, by injecting the above formula into the messages from-variable-to-check node~\eqref{var_comp}, the latter can be written as
\begin{align}
 \BP(\vct{m})_{i \to a }(\vct{x}) &\propto P(\vct{x}) \exp \left( \sum_{b \in \partial i \backslash a} -\half \twonorm{ \mtx{V}_{b \to i}^{-1/2} (\bar{\vct{h}}_b - \widebar{A}_{bi}\vct{x} - \vct{\omega}_{b\to i}) }^2 \right), \nonumber\\
 &\propto P(\vct{x}) \exp\left( - \half \vct{x}^\intercal \left( \sum_{b \in \partial i \backslash a} \widebar{A}_{bi}^2 \mtx{V}_{b \to i}^{-1} \right) \vct{x} + \vct{x}^\intercal \left(\sum_{b \in \partial i \backslash a} \widebar{A}_{bi}\mtx{V}_{b \to i}^{-1} (\bar{\vct{h}}_b - \vct{\omega}_{b \to i}) \right) \right), \nonumber\\
 &\propto P(\vct{x}) \exp(-(\vct{x} - \vct{z}_{i \to a})^\intercal \mtx{\Sigma}_{i \to a}^{-1} (\vct{x} - \vct{z}_{i \to a})/2), \label{var_comp_simplified}
\end{align}
where we denoted the average message and variance by
\begin{align}
\vct{z}_{i \to a} &= \mtx{\Sigma}_{i \to a} \sum_{b \in \partial i \backslash a} \widebar{A}_{bi}\mtx{V}_{b \to i}^{-1} (\bar{\vct{h}}_b - {\bm \omega_{b \to i}}), \label{avg_var_to_check}\\
\mtx{\Sigma}_{i \to a}^{-1} &:= \sum_{b \in \partial i \backslash a} \widebar{A}_{bi}^2 \mtx{V}_{b \to i}^{-1}.\label{variance_var_to_check}
\end{align}

\noindent The combination of the equations (\ref{B_ja}-\ref{variance_var_to_check}) forms the set of \emph{Relaxed Belief Propagation} (RBP) equations: 
\begin{equation}\label{relaxed_BP}
\begin{cases}
\begin{aligned}
\vct{m}_{i \to a} &= \vct{\eta}(\vct{z}_{i \to a}, \mtx{\Sigma}_{i \to a}),\\
\mtx{B}_{i \to a} &= \Diag(\vct{m}_{i \to a}) - \vct{m}_{i \to a} \vct{m}_{i \to a}^\intercal,\\
\vct{z}_{i \to a} &= \mtx{\Sigma}_{i \to a} \sum_{b \in \partial i \backslash a} \widebar{A}_{bi}\mtx{V}_{b \to i}^{-1} (\bar{\vct{h}}_b - \vct{\omega}_{b \to i}), \\
\mtx{\Sigma}_{i \to a}^{-1} &= \sum_{b \in \partial i \backslash a} \widebar{A}_{bi}^2 \mtx{V}_{b \to i}^{-1},\\
\vct{\omega}_{a \to i}  &= \sum_{j \in \partial a \backslash i} \widebar{A}_{aj} \vct{m}_{j \to a},\\
\mtx{V}_{a \to i} &= \sum_{j \in \partial a \backslash i} \widebar{A}_{aj}^2 \mtx{B}_{j \to a}, 
\end{aligned}
\end{cases}
\end{equation}
with 
\begin{equation}\label{vct_eta}
\vct{\eta}(\vct{z}, \mtx{\Sigma}) := \frac{1}{Z(\vct{z}, \mtx{\Sigma})} \sum_{r=1}^d \pi_r \vct{e}_r \exp\left(-\half(\vct{e}_r - \vct{z})^\intercal \mtx{\Sigma}^{-1}(\vct{e}_r - \vct{z})\right),
\end{equation}
where $Z(\vct{z}, \mtx{\Sigma})$ is the normalization constant so that $\one^\intercal \vct{\eta}(\vct{z}, \mtx{\Sigma}) =1$. 
The complexity of the iterative version of these equations is of order at most $\bigo(d^3 n m)$ which is essentially quadratic in $n$. Next, we further reduce the complexity of the iteration to $\bigo(d^3(n+m))$ by showing that it suffices to track the average of the incoming messages at each node. This is due to the fact that the factor graph is dense and its edges are independent.    

\subsection{From Relaxed BP to Approximate Message Passing} 
Let us now derive the equations of the (more efficient) AMP algorithm. We will define a notion of ``total messages" $\vct{m}_{i}, \mtx{B}_{i}$, $\vct{z}_{i}$, $\mtx{\Sigma}_{i}$, $\vct{\omega}_a$, $\mtx{V}_a$ and relate them to one another. The expressions~\eqref{avg_check_to_var}, \eqref{variance_check_to_var}, \eqref{avg_var_to_check}, and~\eqref{variance_var_to_check} defining $\vct{\omega}_{a \to i}$, $\mtx{V}_{a \to i}$, $\vct{z}_{i \to a}$ and $\mtx{\Sigma}_{i \to a}$ respectively involve sums over all the neighbors of the node sending the message except the node receiving the message. We first define $\vct{\omega}_a$, $\mtx{V}_a$ and $\mtx{\Sigma}_{i}$  by adding this last term: 
\begin{align*}
\vct{\omega}_a^t &:=  \sum_{j \in \partial a} \widebar{A}_{aj} \vct{m}_{j \to a}^t = \vct{\omega}_{a \to i}^t + \widebar{A}_{ai} \vct{m}_{i \to a}^t, \\
\mtx{V}_{a}^t &:= \sum_{j \in \partial a} \widebar{A}_{aj}^2 \mtx{B}_{j \to a}^t = \mtx{V}_{a\to i}^t + \widebar{A}_{ai}^2 \mtx{B}_{i \to a}^t,\\
\left(\mtx{\Sigma}_{i}^t\right)^{-1} &:= \sum_{b \in \partial i } \widebar{A}_{bi}^2 \left(\mtx{V}_{b}^t\right)^{-1}.
\end{align*}
where we introduced a time index $t$ to track the iteration count. 
Now we attempt to find a notion of total message $\vct{z}_i^t$ for $\vct{z}_{i \to a}^t$ such that the obtained set of equations becomes self consistent. Once $\vct{z}_i^t$ is found, then we define $\vct{m}_i^{t+1}$ and $\mtx{B}_i^{t+1}$ as $ \vct{\eta}(\vct{z}_i^t,\mtx{\Sigma}_i^t)$ and $\Diag(\vct{\eta}(\vct{z}_i^t,\mtx{\Sigma}_i^t)) - \vct{\eta}(\vct{z}_i^t,\mtx{\Sigma}_i^t) \vct{\eta}(\vct{z}_i^t,\mtx{\Sigma}_i^t)^\intercal$, respectively.  
Since $\mtx{\Sigma}^t_{i \to a} - \mtx{\Sigma}^t_{i} = \bigo(1/n)$ and $\mtx{V}^t_{a \to i} - \mtx{V}^t_{a} = \bigo(1/n)$, we have using~\eqref{avg_var_to_check}
\begin{align*}
\vct{z}_{i \to a}^t &= \mtx{\Sigma}_{i \to a}^t \cdot \sum_{b \in \partial i \backslash a} \widebar{A}_{bi}\left(\mtx{V}_{b \to i}^t\right)^{-1} (\bar{\vct{h}}_b - \vct{\omega}_{b \to i}^t), \\
&\simeq  \mtx{\Sigma}_{i}^t \cdot \sum_{b \in \partial i \backslash a} \widebar{A}_{bi}\left(\mtx{V}_{b}^t\right)^{-1} (\bar{\vct{h}}_b - \vct{\omega}_{b \to i}^t).
\end{align*}
Substituting the expression $\vct{\omega}_{a \to i}^t = \vct{\omega}_{a}^t - \widebar{A}_{ai}\vct{m}_{i \to a}^t$ in the above, we get
\begin{align*}
\vct{z}_{i \to a}^t
&=  \mtx{\Sigma}_{i}^t \cdot \sum_{b \in \partial i \backslash a} \widebar{A}_{bi}\left(\mtx{V}_{b}^t\right)^{-1} (\bar{\vct{h}}_b - \vct{\omega}_{b}^t) + \mtx{\Sigma}_{i}^t \cdot \sum_{b \in \partial i \backslash a} \widebar{A}_{bi}^2 \left(\mtx{V}_{b}^t\right)^{-1} \vct{m}_{i\to b}^t\\
&\simeq  \mtx{\Sigma}_{i}^t \cdot \sum_{b \in \partial i} \widebar{A}_{bi}\left(\mtx{V}_{b}^t\right)^{-1} (\bar{\vct{h}}_b - \vct{\omega}_{b}^t) + \mtx{\Sigma}_{i}^t \cdot \sum_{b \in \partial i } \widebar{A}_{bi}^2 \left(\mtx{V}_{b}^t\right)^{-1} \vct{m}_{i \to b}^t,
\end{align*}
where we also allowed the above sums to run over all neighbors of $i$ since the additional terms are of order $1/\sqrt{n}$ compared to the entire sum which is of order 1.  Now we make the assumption that the messages $\vct{m}^t_{i \to b}$ are approximately equal for all $b \in \partial i$ to a common value $\vct{m}_i^t$, up to error $1/\sqrt{n}$. This assumption is justified by the fact that the graph is dense with equally strong edge weights, so the messages outgoing from every node are equal, up to first order. This simplifies the second term:  
\[\mtx{\Sigma}_{i}^t \cdot \sum_{b \in \partial i } \widebar{A}_{bi}^2 \left(\mtx{V}_{b}^t\right)^{-1} \vct{m}_{i \to b}^t \simeq \mtx{\Sigma}_{i}^t \cdot \sum_{b \in \partial i } \widebar{A}_{bi}^2 \left(\mtx{V}_{b}^t\right)^{-1} \vct{m}_{i}^t = \vct{m}_{i}^t.\]
Based on these approximations, we define 
\[\vct{z}_{i}^t := \mtx{\Sigma}_{i}^t \cdot \sum_{b \in \partial i} \widebar{A}_{bi}\left(\mtx{V}_{b}^t\right)^{-1} (\bar{\vct{h}}_b - \vct{\omega}_{b}^t) + \vct{m}_{i}^t.\]
Now we treat $\vct{\omega}_a^t$. Recall $\vct{\omega}_a^t =  \sum_{j \in \partial a} \widebar{A}_{aj} \vct{m}_{j \to a}^t$,
and $\vct{m}_{j \to a}^{t} = \vct{\eta}(\vct{z}_{j \to a}^{t-1}, \mtx{\Sigma}_{j \to a}^{t-1})$.
We write 
\begin{align*}
\vct{z}_{j \to a}^{t-1} &= \mtx{\Sigma}_{j \to a}^{t-1} \cdot \sum_{b \in \partial j} \widebar{A}_{bj}\left(\mtx{V}_{b \to j}^{t-1}\right)^{-1} (\bar{\vct{h}}_b - \vct{\omega}_{b \to j}^{t-1}) - \mtx{\Sigma}_{j \to a}^{t-1} \cdot \widebar{A}_{aj}\left(\mtx{V}_{a \to j}^{t-1}\right)^{-1} (\bar{\vct{h}}_a - \vct{\omega}_{a \to j}^{t-1}),\\
&\simeq \vct{z}_j^{t-1} - \mtx{\Sigma}_{j \to a}^{t-1} \cdot \widebar{A}_{aj}\left(\mtx{V}_{a}^{t-1}\right)^{-1} (\bar{\vct{h}}_a - \vct{\omega}_{a}^{t-1}).
\end{align*}
The second term is negligible compared to the first one, so we develop a first order Taylor approximation of the function $\vct{\eta}$ in the second term, and obtain
\begin{align*}
\vct{\omega}_a^t &=  \sum_{j \in \partial a} \widebar{A}_{aj} \vct{\eta}(\vct{z}_{j \to a}^{t-1}, \mtx{\Sigma}_{j \to a}^{t-1}), \\
&\simeq  \sum_{j \in \partial a} \widebar{A}_{aj} \left(\vct{\eta}(\vct{z}_{j}^{t-1},\mtx{\Sigma}_{j}^{t-1}) - \frac{\mathrm{d}\vct{\eta}}{\mathrm{d}\vct{z}} (\vct{z}_{j \to a}^{t-1},  \mtx{\Sigma}_{j \to a}^{t-1})  \cdot \mtx{\Sigma}_{j \to a}^{t-1} \cdot\widebar{A}_{aj} (\mtx{V}_a^{t-1})^{-1} (\bar{\vct{h}}_a - \vct{\omega}_a^{t-1})\right),\\
&=  \sum_{j \in \partial a} \widebar{A}_{aj} \vct{m}_{j}^t - \left(\sum_{j \in \partial a} \widebar{A}_{aj}^2 \frac{\mathrm{d}\vct{\eta}}{\mathrm{d}\vct{z}} (\vct{z}_{j \to a}^{t-1},  \mtx{\Sigma}_{j \to a}^{t-1}) \cdot \mtx{\Sigma}_{j \to a}^{t-1}\right)(\mtx{V}_a^{t-1})^{-1}(\bar{\vct{h}}_a - \vct{\omega_a}^{t-1}).
\end{align*}
Based on the expression~\eqref{vct_eta} of $\vct{\eta}$, one can easily check that  
\[\frac{\mathrm{d}\vct{\eta}}{\mathrm{d} \vct{z}}\left(\vct{z}, \mtx{\Sigma}\right) = \left(\Diag(\vct{\eta}(\vct{z}, \mtx{\Sigma})) - \vct{\eta}(\vct{z}, \mtx{\Sigma}) \vct{\eta}(\vct{z}, \mtx{\Sigma})^\intercal\right)\cdot \mtx{\Sigma}^{-1},\]
hence
\[\sum_{j \in \partial a} \widebar{A}_{aj}^2 \frac{\mathrm{d}\vct{\eta}}{\mathrm{d}\vct{z}} (\vct{z}_{j \to a}^{t-1},  \mtx{\Sigma}_{j \to a}^{t-1}) \cdot \mtx{\Sigma}_{j \to a}^{t-1} = \sum_{j \in \partial a} \widebar{A}_{aj}^2 \mtx{B}_{j \to a}^{t} = \mtx{V}_a^{t}.\]
We therefore end up with the following approximate message passing procedure: 
\begin{equation*}
\begin{cases}
\begin{aligned}
\vct{m}_{i}^{t+1} &= \vct{\eta}(\vct{z}_{i}^t, \mtx{\Sigma}_{i}^t),\\
\mtx{B}_{i}^{t+1} &= \Diag(\vct{\eta}(\vct{z}_{i}^t, \mtx{\Sigma}_{i}^t)) - \vct{\eta}(\vct{z}_{i}^t, \mtx{\Sigma}_{i}^t) \vct{\eta}(\vct{z}_{i}^t, \mtx{\Sigma}_{i}^t)^\intercal,\\
\mtx{\Sigma}_{i}^t ~~~&= \left(\sum_{b \in \partial i } \widebar{A}_{bi}^2 \left(\mtx{V}_{b}^t\right)^{-1}\right)^{-1},\\
\vct{z}_{i}^t ~~~&= \vct{m}_{i}^t + \mtx{\Sigma}_{i}^t \cdot \sum_{b \in \partial i} \widebar{A}_{bi} \left(\mtx{V}_{b}^t\right)^{-1} (\bar{\vct{h}}_b - \vct{\omega}^t_{b}), \\
\vct{\omega}_{a}^t  ~~~&= \sum_{j \in \partial a} \widebar{A}_{aj} \vct{m}_{j}^t - \mtx{V}_{a}^t\left(\mtx{V}_{a}^{t-1}\right)^{-1}(\bar{\vct{h}}_a - \vct{\omega}_a^{t-1}),\\
\mtx{V}_{a}^t ~~~&= \sum_{j \in \partial a} \widebar{A}_{aj}^2 \mtx{B}_{j}^t.
\end{aligned}
\end{cases}
\end{equation*}
This is rearranged to the AMP algorithm displayed in Section~\ref{subsec:AMP}, with the notation $\hat{\vct{x}}_i^t$ replacing $\vct{m}_i^t$.
%%%%%%%%%%%%%%%%%%%%%%%%%%%%%%%%%%%%%%%%%%%%%%
% SE

\section{State Evolution equations} 
 \label{sxn:SE}
We derive the state evolution equations from the Relaxed Belief Propagation (RBP) equations~\eqref{relaxed_BP}. 
Let $\mtx{M}_t = \frac{1}{n}\sum_{i=1}^n \vct{m}_i^t\vct{x}_i^{*\intercal}$ and $\mtx{Q}_t = \frac{1}{n}\sum_{i=1}^n \vct{m}_i^t\vct{m}_i^{t\intercal}$. As we argued in the previous section, we can redefine $\mtx{M}_t$ and $\mtx{Q}_t$ by substituting $\vct{m}_{i}^t$ by $\vct{m}_{i \to a}^t$ at the cost of an asymptotically vanishing error. In this section, we drop the time indices to lighten the notation.
We expect the variance parameters $\mtx{V}_{a \to i}$ in RBP to be concentrated about a constant:
\[\E[\mtx{V}_{a \to i}] \simeq \sum_{j \neq i} \E[ \widebar{A}_{aj}^2] \mtx{B}_{j \to a}  = \frac{1}{n}\alpha(1-\alpha)\sum_{j \neq i} \mtx{B}_{j \to a} = \alpha(1-\alpha)\mtx{R}, \]
with $\mtx{R} := \frac{1}{n} \sum_{j } \mtx{B}_{j \to a}$.
A calculation of the second moment of $\mtx{V}_{a \to i}$ reveals that it is equal to the expectation of $\mtx{V}_{a \to i}$ plus a lower order term. Therefore we can safely assume that the quantities $\mtx{V}_{a \to i}$ are essentially constant and equal to $\alpha(1-\alpha) \mtx{R}$. 
Next, we deal with $\mtx{\Sigma}_{i \to a}$. By assuming approximate independence of $\widebar{A}_{bi}$ and $\mtx{V}_{b \to i}$, we get 
\[\E\left[\mtx{\Sigma}_{i \to a}^{-1}\right] = \sum_{b \neq a} \E\left[\widebar{A}_{bi}^2\right] \E\left[\mtx{V}_{b \to i}^{-1}\right]  =  \frac{1}{n} \alpha(1-\alpha) \sum_{b \neq a} \frac{\mtx{R}^{-1}}{\alpha(1-\alpha)}\simeq \kappa \mtx{R}^{-1}.\]
We then make the approximation $\mtx{\Sigma}_{i \to a}^{-1} \simeq \E[\mtx{\Sigma}_{i \to a}^{-1}]$, i.e.\ $\mtx{\Sigma}_{i \to a} \simeq \kappa^{-1}\mtx{R}$.
Next, we turn our attention to $\vct{z}_{i \to a}$:
\begin{align*}
\vct{z}_{i \to a} &= \mtx{\Sigma}_{i \to a} \cdot \sum_{b \neq a} \widebar{A}_{bi}\mtx{V}_{b \to i}^{-1} (\bar{\vct{h}}_b - {\bm \omega_{b \to i}})\\
&\simeq \frac{1}{\kappa \alpha(1-\alpha)} \sum_{b \neq a} \widebar{A}_{bi}(\bar{\vct{h}}_b - {\bm \omega_{b \to i}}). 
\end{align*}

\noindent Now using $\vct{\omega}_{b \to i}  = \sum_{j \neq i} \widebar{A}_{bj} \vct{m}_{j \to b}$ and $\bar{\vct{h}}_b = \sum_{ j =1}^n \widebar{A}_{bj}\vct{x}^*_j$, we get 
\[\vct{z}_{i \to a} \simeq \frac{1}{\kappa \alpha(1-\alpha)} \sum_{b \neq a} \widebar{A}_{bi}\left(\sum_{ j \neq i} \widebar{A}_{bj}(\vct{x}^*_j -\vct{m}_{j \to a}) + \widebar{A}_{bi}\vct{x}^*_i\right).\] 
The inner sum in the above expression involves $n$ weakly independent terms, so we expect a central limit theorem to hold. Therefore the only relevant quantities are the expectation and the variance of $\vct{z}$: $\E[\vct{z}_{i \to a}] = \vct{x}^*_i$, and 
\begin{align*}
\E[(\vct{z}_{i \to a}-\vct{x}^*_i)(\vct{z}_{i \to a}-\vct{x}^*_i)^\intercal] &= \frac{1}{(\kappa \alpha(1-\alpha))^2}\sum_{b \neq a} \sum_{j \neq i} \sum_{b' \neq a} \sum_{j' \neq i} \E[\widebar{A}_{bi}\widebar{A}_{b'i}]\E[\widebar{A}_{bj}\widebar{A}_{bj'}] \\
& \hspace{2in} \times (\vct{x}^*_j -\vct{m}_{j \to a})(\vct{x}^*_j -\vct{m}_{j \to a})^\intercal\\
 &= \frac{1}{(\kappa \alpha(1-\alpha))^2}\sum_{b \neq a} \sum_{j \neq i} \frac{(\alpha(1-\alpha))^2}{n^2} (\vct{x}^*_j -\vct{m}_{j \to a})(\vct{x}^*_j -\vct{m}_{j \to a})^\intercal\\
 &= \kappa^{-2} \frac{m}{n} \frac{1}{m}\sum_{b \neq a} \frac{1}{n} \sum_{j \neq i}  (\vct{x}^*_j -\vct{m}_{j \to a})(\vct{x}^*_j -\vct{m}_{j \to a})^\intercal\\
 &\simeq \kappa^{-1} (\mtx{D} - \mtx{M} -\mtx{M}^\intercal + \mtx{Q}),
\end{align*}
with $\mtx{D} = \frac{1}{n} \sum_{i=1}^n \vct{x}_i^* \vct{x}_i^{*\intercal} = \Diag(\vct{\pi})$. 
Hence, we define \[\mtx{X} := \kappa^{-1}(\mtx{D} - \mtx{M} -\mtx{M}^\intercal + \mtx{Q}).\]
Therefore we have made the assumption that $\vct{z}_{i \to a} \sim \normal(\vct{x}^*_i, \mtx{X})$. Next, we assume that the $\vct{z}_{i \to a}$ are ``independent enough" that a law of large numbers holds in limit $n \to \infty$, $m/n \to \kappa$: 
 \[\frac{1}{n} \sum_{i: \vct{x}^*_i = \vct{e}_r} \vct{m}_{i \to a} = \frac{1}{n} \sum_{i:\vct{x}^*_i = \vct{e}_r} \vct{\eta}(\vct{z}_{i \to a},\mtx{\Sigma}_{i \to a}) \simeq \pi_r \E_{\vct{g}}\left[\vct{\eta}(\vct{e}_r + \mtx{X}^\half\vct{g}, \kappa^{-1}\mtx{R})\right],\]
and
\[\frac{1}{n} \sum_{i: \vct{x}^*_i = \vct{e}_r} \vct{m}_{i \to a} \vct{m}_{i \to a}^\intercal \simeq \pi_r \E_{\vct{g}}\left[\vct{\eta}(\vct{e}_r + \mtx{X}^\half\vct{g}, \kappa^{-1}\mtx{R}) \cdot \vct{\eta}(\vct{e}_r + \mtx{X}^\half\vct{g}, \kappa^{-1}\mtx{R})^\intercal \right],\]
for all $r \in \{1,\cdots,d\}$, with $\vct{g} \sim \normal(\mtx{0},\mtx{I})$. Plugging the above into $\mtx{M}$ and $\mtx{Q}$ yields 
\begin{align*}
\mtx{M} &= \frac{1}{n} \sum_{i=1}^n \vct{\eta}(\vct{x}^*_i + \mtx{X}^\half\vct{g}, \kappa^{-1}\mtx{R}) \vct{x}_i^{*\intercal}, \\
&\simeq  \sum_{r=1}^d \pi_r\E_{\vct{g}}\left[\vct{\eta}(\vct{e}_r + \mtx{X}^\half\vct{g}, \kappa^{-1}\mtx{R})\right] \vct{e}_r^\intercal,\\
\mtx{Q} &=  \frac{1}{n} \sum_{i=1}^n \vct{\eta}(\vct{x}^*_i + \mtx{X}^\half\vct{g}, \kappa^{-1}\mtx{R}) \cdot \vct{\eta}(\vct{x}^*_i + \mtx{X}^\half\vct{g}, \kappa^{-1}\mtx{R})^\intercal,\\
&\simeq   \sum_{r=1}^d \pi_r \E_{\vct{g}}\left[\vct{\eta}(\vct{e}_r + \mtx{X}^\half\vct{g}, \kappa^{-1}\mtx{R}) \cdot \vct{\eta}(\vct{e}_r + \mtx{X}^\half\vct{g}, \kappa^{-1}\mtx{R})^\intercal \right].
\end{align*}

Finally, it remains to find an expression for $\mtx{R}$. Recall $\mtx{B}_{i \to a} =  \Diag(\vct{m}_{i \to a}) - \vct{m}_{i \to a} \vct{m}_{i \to a}^\intercal$.
Averaging over $i$ and using the assumed concentration of the messages $\vct{m}_{i \to a}$ yields 
\begin{align*}
\mtx{R} = \frac{1}{n} \sum_{i=1}^n \mtx{B}_{ i\to a} &\simeq \Diag\left(\sum_{r=1}^d \pi_r\E_{\vct{g}}\left[\vct{\eta}(\vct{e}_r + \mtx{X}^{\half}\vct{g}, \kappa^{-1}\mtx{R})\right]\right) - \mtx{Q},\\
&= \Diag(\mtx{Q}\one) - \mtx{Q}.
\end{align*}
To sum up, we get a system of self-consistent equations in $\mtx{M}_t$, $\mtx{Q}_t$, $\mtx{X}_t$ and $\mtx{R}_t$:
\begin{equation*}
 \begin{cases}
 \begin{aligned}
\mtx{M}_{t+1} &= \sum_{r=1}^d \pi_r\E_{\vct{g}}\left[\vct{\eta}(\vct{e}_r + \mtx{X}_{t}^{\half}\vct{g}, \kappa^{-1}\mtx{R}_{t})\right] \cdot \vct{e}_r^\intercal, \\
\mtx{Q}_{t+1} &= \sum_{r=1}^d \pi_r\E_{\vct{g}}\left[\vct{\eta}(\vct{e}_r + \mtx{X}_{t}^{\half}\vct{g}, \kappa^{-1}\mtx{R}_{t}) \cdot \vct{\eta}(\vct{e}_r + \mtx{X}_{t}^{\half}\vct{g}, \kappa^{-1}\mtx{R}_{t})^\intercal \right],\\
\mtx{X}_t ~~~&= \kappa^{-1}(\mtx{D} - \mtx{M}_t -\mtx{M}_{t}^{\intercal} + \mtx{Q}_t),\\
\mtx{R}_t ~~~&= \Diag(\mtx{Q}_t\one) - \mtx{Q}_t.
\end{aligned}
\end{cases}
\end{equation*}
This set of equations constitute the State Evolution equations.

\end{document}